\documentclass[10pt,conference]{IEEEtran}
\usepackage[keeplastbox]{flushend}
\usepackage[utf8]{inputenc}
\usepackage{amsmath}
\usepackage{amsthm}
\usepackage{mathtools}
\usepackage{amsfonts}
\usepackage{hyperref}
\usepackage{url}
\usepackage{color}
\usepackage{tikz}
\usepackage{stackengine}
\usepackage{multiaudience}
  \SetNewAudience{conference}
  \SetNewAudience{arxiv}
\newcommand\blankfootnote[1]{%
\begin{NoHyper}
  \let\svthefootnote\thefootnote%
  \let\thefootnote\relax\footnotetext{#1}%
  \let\thefootnote\svthefootnote%
\end{NoHyper}
}

\newcommand{\Real}{\mathop{\mathrm{Re}}\nolimits}

\newcommand{\E}{{\mathbb{E}}}

\newcommand{\eps}{{\varepsilon}}        %%%%%%%%%%%%%%%%%%%%%%%%%%%%%%%%%%%

\newcommand{\fr}{{\mathcal F}}

\newcommand{\sign}{\operatorname{sign}}

\newcommand{\NN}{{\mathbb{N}}}

\newcommand{\complex}{\mathbb{C}}
\newcommand{\reals}{\mathbb{R}}

\newcommand{\Bs}{{\mathcal{B}}}

\newcommand{\Ss}{{\mathcal{S}}}

\newcommand{\subgaussnorm}[1]{{\tau\left({#1}\right)}}
\newcommand{\subexpnorm}[1]{{\theta\left({#1}\right)}}
\newtheorem{theorem}{Theorem}

\newtheorem{lemma}{Lemma}
\newtheorem{cor}{Corollary}

\newtheorem{remark}{Remark}
\newtheorem{definition}{Definition}
\newtheorem{example}{Example}

\title{Distributed Approximation of Functions over Fast Fading Channels with Applications to Distributed Learning and the Max-Consensus Problem}

\author{\IEEEauthorblockN{ Igor Bjelakovi\'c\IEEEauthorrefmark{1}\IEEEauthorrefmark{2}, %
    Matthias Frey\IEEEauthorrefmark{1}\IEEEauthorrefmark{2} and %
    S\l awomir~Sta\'{n}czak\IEEEauthorrefmark{2}\IEEEauthorrefmark{3}\\ %
    } %
  \IEEEauthorblockA{
  \IEEEauthorrefmark{2}Technische Universität Berlin, Germany
  and
  \IEEEauthorrefmark{3}Fraunhofer Heinrich Hertz Institute, Berlin,
   Germany
  }%
  \vspace*{-2em} }%%
  
\begin{document}
\maketitle
\blankfootnote{
This work was supported by the German Federal Ministry of Education and Research (BMBF) under grant 16KIS0605 and by the German Research Foundation (DFG) within their priority program SPP 1914 ``Cyber-Physical Networking''.

\IEEEauthorrefmark{1} The first two authors contributed equally to this work.
}
\begin{abstract}
  In this work, we consider the problem of distributed approximation
  of functions over multiple-access channels with additive noise. In
  contrast to previous works, we take fast fading into account and
  give explicit probability bounds for the approximation error
  allowing us to derive bounds on the number of channel uses that are
  needed to approximate a function up to a given approximation
  accuracy. Neither the fading nor the noise process is limited to
  Gaussian distributions. Instead, we consider sub-gaussian random
  variables which include Gaussian as well as many other distributions
  of practical relevance. The results are motivated by and have
  immediate applications to a) computing predictors in models for
  distributed machine learning and b) the max-consensus problem in
  ultra-dense networks.
\end{abstract}

\section{Introduction}
Massive wireless sensor networks with thousands of sensors are expected
to enable many important 5G applications including mobile health care,
environment monitoring, smart transportation and smart agriculture to
name a few. Efficient and reliable data collection in such massive
communication scenarios requires fundamentally new approaches to the
problem of massive access with many sensors competing for access to
scarce wireless resources.
An important entry point for improvements is that it is often not
necessary to reconstruct all the individual transmitted data, but
rather some function of them~\cite{nazer2007computation}. In addition,
many applications do not require the exact function value but can
instead work with a noisy version of it as long as the noise is bounded or
otherwise controlled. In this work, we focus on a particular class of
objective functions and propose a method for approximating them over
fast fading, noisy channels. One particular application we have in
mind is a distributed computation of the estimator function of machine
learning models. Considering the recent interest in machine learning,
this can be expected to become an increasingly important problem in
future wireless networks.

Generally, we expect functions of the form
\begin{equation*}
f(s_1, \ldots, s_K)= F\left(  \sum_{k=1}^K f_k(s_k)     \right ),
\end{equation*}
called \emph{nomographic functions},
to be amenable to distributed approximation over a wireless channel in
which a superposition of signals results in a noisy sum of the
transmitted signals to arrive at the receiver, and in fact it turns
out that every multivariate real function $f$ has such a
representation~\cite{buck1976approximate}. However, even extremely
weak noise in the individual components can have an unpredictable
impact on the overall error if such representations are
used. Therefore, it is necessary to introduce additional requirements
on $f_1, \dots, f_K$ and $F$. It is known~\cite{kolmogorov1957representation} that every continuous function mapping from $[0,1]^K$ to $\reals$ can be represented as a sum of no more than $2K+1$ nomographic functions with continuous representations. Another result worth noting in this context is that nomographic functions with continuous representations are nowhere dense in the space of continuos functions~\cite{buck1982nomographic} and thus the representation as a sum of nomographic functions is really necessary.
However, even with suitable continuous representations available, it is hard to control the impact of the channel noise. We therefore consider a different
class of functions and use $\fr_{\textrm{\textrm{mon}}}$ to denote
this class. Although there are functions of practical interest that
are not in $\fr_{\textrm{\textrm{mon}}}$, we observe that many important
functions belong to this class. We show how to approximate these
functions in a distributed fashion in a massive access scenario with
fast fading and additive noise. The fading and noise distributions are
assumed to be sub-gaussian, 
which includes Gaussian distributions as a special case as well as many other
practical distributions.

Distributed computation of functions has been introduced
in~\cite{nazer2007computation} with applications in network coding,
but in contrast to this approach of exactly and repeatedly computing instances of the
same discrete function with arguments drawn from a known random
distribution, we focus on approximate one-shot computation
of analog functions with arbitrary arguments. This means that we do
not have a computation rate, but instead an approximation error and an
associated number of channel uses which is uniform in the transmitted data, and can thus be arbitrary. An assumption that it follows a probability distribution is not necessary. We revisit the approach from~\cite{kiril}, modify it slightly, and provide a detailed
theoretical analysis of the approximation error. Other approaches to
and applications of the distributed approximation of nomographic
functions appeared
in~\cite{goldenbaum2013robust,goldenbaum2013harnessing,goldenbaum2014nomographic,goldenbaum2016harnessing}.

Our main contributions in this work are
\begin{enumerate}
 \item a detailed technical analysis of a method of distributed approximation of functions in $\fr_{\textrm{mon}}$ in a multiple-access setting with fast fading and additive noise,
 \item the treatment of sub-gaussian fading and noise, generalizing
   the Gaussian case so as to accommodate many fading and noise
   distributions that occur in practice,
 \item applications of these techniques to a subclass of machine learning models in Section~\ref{sec:ml} and to a highly scalable max-consensus protocol for ultra-dense networks in Section~\ref{sec:max-consensus}.
\end{enumerate}

\showto{conference}{Some proofs are omitted due to space constraints. They can be found, however, in the extended version of this work~\cite{arxivpaper}.}
%
%%%%%%%%%%%%%%%%%%%%%%%%%%%%%%%%%%%%%%%%%%%%%%%%%%%%%%%%%%%%%%%%%%%
%
%						NETWORK MODEL
%
%%%%%%%%%%%%%%%%%%%%%%%%%%%%%%%%%%%%%%%%%%%%%%%%%%%%%%%%%%%%%%%%%%%
%

%
%%%%%%%%%%%%%%%%%%%%%%%%%%%%%%%%%%%%%%%%%%%%%%%%%%%%%%%%%%%%%%%%%%%
%
%						---------------
%
%%%%%%%%%%%%%%%%%%%%%%%%%%%%%%%%%%%%%%%%%%%%%%%%%%%%%%%%%%%%%%%%%%%
%

%%%%%%%%%%%%%%%%%%%%%%%%%%%%%%%%%%%%%%%%%%%%%%%%%%%%%%%%%%%%%%%%%%%
%
%                                           SECTION: SYSTEM MODEL AND PROBLEM STATEMENT
%
%%%%%%%%%%%%%%%%%%%%%%%%%%%%%%%%%%%%%%%%%%%%%%%%%%%%%%%%%%%%%%%%%%%
\section{System Model and Problem Statement}     

%%%%%%%%%%%%%%%%%%%%%%%%%%%%%%%%%%%%%%%%%%%%%%%%%%%%%%%%%%%%%%%%%%
%
%                                             INTRO SUB-GAUSSIAN RANDOM VARIABLES
%
%%%%%%%%%%%%%%%%%%%%%%%%%%%%%%%%%%%%%%%%%%%%%%%%%%%%%%%%%%%%%%%%%%%
\subsection{Sub-Gaussian Random Variables}
We begin with a short overview of the relevant definitions and properties of sub-gaussian random variables.
\showto{arxiv}{More on this topic can be found in Section \ref{sec:sub-exp-sub-gauss} and in \cite{buldygin,wainwright,vershynin}.}
\showto{conference}{More on this topic can be found in \cite{arxivpaper,buldygin,wainwright,vershynin}.}

For a random variable $X$, we define\footnote{Note that other norms on the space of sub-gaussian random variables that appear in the literature are equivalent to $\subgaussnorm{\cdot}$ (see, e.g.,~\cite{buldygin}). The particular definition we choose here matters, however, because we want to derive results in which no unspecified constants appear.}
\begin{multline}\label{eq:sub-gauss-norm-def-first}
\subgaussnorm{X} :=  \inf \Big\{t > 0: \forall \lambda \in \reals \\ \mathbb{E}\exp \left( \lambda (X - \mathbb{E} X) \right)   \le \exp \left( \lambda^2 t^2 / 2 \right)      \Big\}.
\end{multline}
 $X$ is called a sub-gaussian random variable if $\subgaussnorm{X}<\infty $. The function $\subgaussnorm{\cdot}$ defines a semi-norm on the set of sub-gaussian random variables \cite[Theorem 1.1.2]{buldygin}, i.e., it is absolutely homogeneous, satisfies the triangle inequality, and is non-negative. $ \subgaussnorm{X}=0$ does not necessarily imply $X=0$ unless we identify random variables which are equal almost everywhere.
Examples of sub-gaussian random variables include Gaussian and bounded random variables.
%%%%%%%%%%%%%%%%%%%%%%%%%%%%%%%%%%%%%%%%%%%%%%%%%%%%%%%%%%%%%%%%%%%
%
%.                                          SYSTEM
%
%%%%%%%%%%%%%%%%%%%%%%%%%%%%%%%%%%%%%%%%%%%%%%%%%%%%%%%%%%%%%%%%%%%
\subsection{System Model}
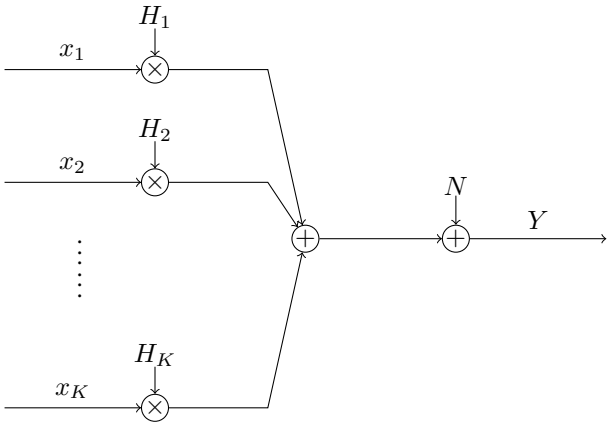
\begin{figure}
\begin{tikzpicture}
\coordinate                      (inK)          at (0,0);
\coordinate                      (in2)          at (0,3);
\coordinate                      (in1)          at (0,4.5);
\node[circle,inner sep=0pt,draw] (multK)        at (2,0)    {$\times$};
\node[circle,inner sep=0pt,draw] (mult2)        at (2,3)    {$\times$};
\node[circle,inner sep=0pt,draw] (mult1)        at (2,4.5)  {$\times$};
\node                            (vdots)        at (1,1.85) {\Shortstack{. . . . . .}};
\node[inner sep=0pt]             (fadingK)      at (2,.7)   {$H_K$};
\node[inner sep=0pt]             (fading2)      at (2,3.7)  {$H_2$};
\node[inner sep=0pt]             (fading1)      at (2,5.2)  {$H_1$};
\coordinate                      (cornerK)      at (3.5,0);
\coordinate                      (corner2)      at (3.5,3);
\coordinate                      (corner1)      at (3.5,4.5);
\node[circle,inner sep=0pt,draw] (plus)         at (4,2.25) {$+$};
\node[circle,inner sep=0pt,draw] (plusnoise)    at (6,2.25) {$+$};
\node[inner sep=0pt]             (noise)        at (6,2.95) {$N$};
\coordinate                      (out)          at (8,2.25);

\draw[->] (inK) -- (multK) node[midway,above] {$x_K$};
\draw[->] (in2) -- (mult2) node[midway,above] {$x_2$};
\draw[->] (in1) -- (mult1) node[midway,above] {$x_1$};

\draw[->] (fadingK) -- (multK);
\draw[->] (fading2) -- (mult2);
\draw[->] (fading1) -- (mult1);

\draw[->] (multK) -- (cornerK) -- (plus);
\draw[->] (mult2) -- (corner2) -- (plus);
\draw[->] (mult1) -- (corner1) -- (plus);

\draw[->] (plus) -- (plusnoise);
\draw[->] (noise) -- (plusnoise);

\draw[->] (plusnoise) -- (out) node[midway,above] {$Y$};
\end{tikzpicture}
\caption{Channel model.}
\label{fig:channel}
\end{figure}

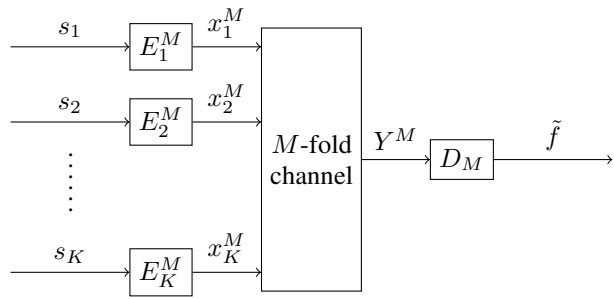
\begin{figure}
\begin{tikzpicture}
\coordinate                      (inK)          at (0,0);
\coordinate                      (in2)          at (0,2);
\coordinate                      (in1)          at (0,3);
\node[rectangle,draw]            (encK)         at (2,0) {$E_K^M$};
\node[rectangle,draw]            (enc2)         at (2,2) {$E_2^M$};
\node[rectangle,draw]            (enc1)         at (2,3) {$E_1^M$};
\node[rectangle,minimum height=3.5cm,align=center,draw] (channel) at (4,1.5) {$M$-fold\\channel};
\node[rectangle,draw]            (dec)          at (6,1.5) {$D_M$};
\coordinate                      (out)          at (8,1.5);
\node                            (vdots)        at (.8,1.2) {\Shortstack{. . . . . .}};

\draw[->] (inK) -- (encK) node[midway,above] {$s_K$};
\draw[->] (in2) -- (enc2) node[midway,above] {$s_2$};
\draw[->] (in1) -- (enc1) node[midway,above] {$s_1$};

\draw[->] (encK) -- (encK-|channel.west) node[midway,above] {$x_K^M$};
\draw[->] (enc2) -- (enc2-|channel.west) node[midway,above] {$x_2^M$};
\draw[->] (enc1) -- (enc1-|channel.west) node[midway,above] {$x_1^M$};

\draw[->] (channel) -- (dec) node[midway,above] {$Y^M$};
\draw[->] (dec) -- (out) node[midway,above] {$\tilde{f}$};
\end{tikzpicture}
\caption{System model.}
\label{fig:system}
\end{figure}

We consider the following channel model with $K$ transmitters and one receiver, depicted in Fig.~\ref{fig:channel}: For $m=1,\ldots, M$, the channel output at the $m$-th channel use is given by
\begin{equation}\label{eq:channel-model}
Y(m)=\sum_{k=1}^{K} H_k(m)x_k(m)+ N(m),
\end{equation}
where:
\begin{itemize}
\item $x_k(m)\in \complex$ are transmit symbols. We assume a peak power constraint $|x_k(m)  |^2 \le P$ for $k=1,\ldots, K$ and $m=1,\ldots, M$.
\item $H_k(m)$, $k=1\ldots, K$, $m=1,\ldots, M$, are independent complex-valued random variables such that for every $m=1, \ldots ,M$ and $k=1,\ldots,K$, the real part $H_k^r (m)$ and the imaginary part $H_k^i (m)$  of $H_k(m)$ are independent, sub-gaussian random variables
with mean zero and variance $1$. Moreover, we assume that there is a $\sigma_F\ge 0$ with
\begin{equation}\label{eq:bound-fading}
 \max\left( \tau (H^r_k(m)),  \tau (H^i_k(m))\right) \le \sigma_F
\end{equation}
for all $k=1, \ldots, K$ and $m=1,\ldots ,M$.
\item $N(m)$, $m=1,\ldots, M$, are independent complex-valued random variables. We assume that the real and imaginary parts $N^{r} (m),N^{i }(m)$ of $N(m)$  are independent sub-gaussian random variables  with mean zero for $m=1,\ldots, M$ and that there is a real number $\sigma_N\ge 0$ such that
\begin{equation}\label{eq:bound-noise}
 \max\left( \tau (N^{r}(m) ), \tau (N^{i}(m) ) \right) \le \sigma_N,
\end{equation}
for $m=1,\ldots, M$.
\item We assume that $N(m)$ and $H_k(m)$, $k=1,\ldots K$, $m=1,\ldots ,M$, are independent.
\end{itemize}
\begin{remark}
If for all $m$, $N^{r} (m),N^{i }(m)$ is distributed according to $\mathcal{N}(0, \sigma^2)$, i.e., $N(m)$ is a circularly symmetric gaussian random variable, then a direct computation of
the moment generating function in (\ref{eq:sub-gauss-norm-def-first}) shows that $\sigma_N$ in (\ref{eq:bound-noise}) can be replaced with the variance $\sigma$ of $N^{r} (m),N^{i }(m) $.

The same holds for the fading random variables $H_k^r (m)$ and $H_k^i (m)$. If they follow the standard normal distribution, $\sigma_F$ in (\ref{eq:bound-fading}) can be replaced with $1$.
\end{remark}
%%%%%%%%%%%%%%%%%%%%%%%%%%%%%%%%%%%%%%%%%%%%%%%%%%%%%%%%%%%%%%%%%%%%
%
%                                                  PRE- AND POSTPROCESSING: DEFINITIONS
%
%%%%%%%%%%%%%%%%%%%%%%%%%%%%%%%%%%%%%%%%%%%%%%%%%%%%%%%%%%%%%%%%%%%
\subsection{Distributed Approximation of Functions}
Our goal is to approximate functions  $f: \Ss_1\times \ldots \times \Ss_K\to\reals$  in a distributed setting. The sets $\Ss_1,\ldots \Ss_K\subseteq \reals $ are assumed to be closed and endowed with their natural Borel $\sigma$-algebras 
$\Bs(\Ss_1),\ldots ,\Bs(\Ss_K)$, and we consider the  product $\sigma$-algebra $\Bs (\Ss_1)\otimes \ldots \otimes \Bs(\Ss_K)$ on the set $ \Ss_1\times \ldots \times \Ss_K $. Furthermore, the functions  $f: \Ss_1\times \ldots \times \Ss_K\to\reals$ under consideration are assumed to be measurable in what follows.

An admissible distributed function approximation scheme for $f: \Ss_1\times \ldots \times \Ss_K\to\reals$ for $M$ channel uses, depicted in Fig.~\ref{fig:system}, is a pair $(E^M, D^M)$, consisting of:

\begin{enumerate}
\item A pre-processing function $E^M = (E_1^M, \dots, E_K^M)$, where each $E_k^M$ is of the form
\[E_k^M(s_k)=(x_k(m, s_k, U_k(m) ))_{m=1}^{M}\in \complex^{M} \]
with random variables $U_k(1), \ldots, U_k(M)$ and a measurable map 
\[(s_k, t_1, \ldots, t_M)\mapsto  (x_k(m, s_k,  t_m ))_{m=1}^{M}\in \complex^{M}.\]
The encoder $E_k^M$ is subject to the peak power constraint $ |x_k(m, s_k, U_k(m))  |^2 \le P$ for all $k=1,\ldots , K$ and $m=1,\ldots , M$.
\item A post-processing function $D^M$: The receiver is allowed to apply a measurable recovery function $D^M: \complex ^M\to \reals$ upon observing the output of the channel.
\end{enumerate}
So in order to approximate $f$, the transmitters apply their pre-processing maps to
\[(s_1,\ldots, s_K)\in \Ss_1\times \ldots \times \Ss_K\]
resulting in $E_1^M (s_1), \ldots, E_K^M(s_K)$ which are sent over the channel. 
The receiver observes the output of the channel and applies the recovery map $D^M$. The whole process defines an estimate $\tilde{f}$ of $f$.

Let $\varepsilon,\delta \in (0,1)$ and $f: \Ss_1\times \ldots \times \Ss_K\to\reals $ be given. We say that $f$ is $\varepsilon$-approximated after $M$ channel uses with confidence level $\delta$ if there is an approximation scheme
$(E^M,D^M)$ such that the resulting estimate $\tilde{f}$ of $f$ satisfies
\begin{equation}\label{eq:eps-delta-approx}
\mathbb{P}( |    \tilde{f} (s^K)- f(s^K)       |\ge \eps      )\le \delta
\end{equation}
for all $s^K:= (s_1, \ldots , s_K) \in \Ss_1\times \ldots \times \Ss_K $.
Let $M(f, \varepsilon, \delta)$ denote the smallest non-negative integer such that there is an approximation scheme $(E^M, D^M)$ for $f$ satisfying (\ref{eq:eps-delta-approx}). We call $M(f, \varepsilon, \delta)$ the communication cost for approximating a function $f$
with accuracy $\varepsilon$ and confidence $\delta$.

%%%%%%%%%%%%%%%%%%%%%%%%%%%%%%%%%%%%%%%%%%%%%%%%%%%%%%%%%%%%%%%%%%%
%
%						CLASS OF FUNCTIONS
%
%%%%%%%%%%%%%%%%%%%%%%%%%%%%%%%%%%%%%%%%%%%%%%%%%%%%%%%%%%%%%%%%%%%
%
\subsection{The class of functions to be approximated}
We set for $k=1,\ldots , K$
\begin{equation}
\fr_{k, \infty}:= \{ f:\Ss_k \to \reals : f \textrm{ is measurable and bounded}     \}.
\end{equation}
A measurable function  $f:\Ss_1\times \ldots \times \Ss_K \to \reals$ is called a \emph{generalized linear function} if there are bounded measurable functions $f_k \in \fr_{k,\infty}$,
$k=1, \ldots , K$, with
\begin{equation}
f(s_1, \ldots , s_K)= \sum_{k=1}^K f_k(s_k),
\end{equation}
for all $(s_1,\ldots , s_K)\in \Ss_1 \times \ldots \times \Ss_K$. The set of generalized linear functions from $\Ss_1 \times \ldots \times \Ss_K\to \reals$ is denoted by $\fr_{K, \textrm{lin}}$.
Our main object of interest will be the following class of functions.
%%%%%%%%%%%%%%%%%%%%%%%%%%%%%%%%%%%%%%%%%%%
%
%                             DEFINITION:
%
%%%%%%%%%%%%%%%%%%%%%%%%%%%%%%%%%%%%%%%%%%%%
\begin{definition}
\label{def:Fmon}
A measurable function $f: \Ss_1\times \ldots \times \Ss_K\to \reals$ is said to belong to $\fr_{\textrm{\textrm{mon}}}$ if there exist
$f_k\in \fr_{k, \infty}$, $k=1,\ldots , K$, a measurable set $D\subseteq \reals$ with the property $f_1(\Ss_1)+\ldots + f_K(\Ss_K)\subseteq D$, a measurable function $F:D\to \reals$ such that
for all $(s_1, \ldots , s_K)\in \Ss_1\times \ldots \times \Ss_K$ we have
\begin{equation}\label{eq:nomographic-def}
f(s_1, \ldots, s_K)= F\left(  \sum_{k=1}^K f_k(s_k)     \right ),
\end{equation}\label{eq:monotonic}
and there is a strictly increasing function $\Phi : [0, \infty) \to [0, \infty)$  with $\Phi(0)=0$ and
\begin{equation}\label{eq:monotone-domination}
| F(x)-F(y)|\le \Phi( | x-y | )
\end{equation}
for all $x,y \in D$. We call the function $\Phi$ an \emph{increment majorant} of $f$.
\end{definition}
%%%%%%%%%%%%%%%%%%%%%%%%%%%%%%%%%%%%%%%%%%%%%%%%
Some examples of functions in $\fr_{\textrm{mon}}$ are:
\begin{enumerate}
\item Obviously, all $f\in\fr_{K, \textrm{lin}}$ belong to $ \fr_{\textrm{mon}}$.
\item For any $f\in\fr_{K, \textrm{lin}}$ and $B$-Lipschitz function $F:\reals \to \reals$ we have $F\circ f\in \fr_{\textrm{mon}}$ with
$\Phi: [0, \infty) \to [0, \infty)$, $x\mapsto Bx$.
\item For any $p\ge 1$ and $\Ss_1, \ldots , \Ss_K$ compact, $|| \cdot ||_p \in \fr_{\textrm{mon}}$. In this example we have $f_k(s_k)=| s_k |^p$, $k=1,\ldots , K$,
$F: [0, \infty) \to [0, \infty)$, $x \mapsto x^{\frac{1}{p}}$, and $F=\Phi$.\\
This can be seen as follows. We have to show that for all nonnegative $x,y\in \reals$ and $p\ge 1$ we have
\begin{equation}
\label{eq:norm-mon}
| x^{\frac{1}{p}}   -y^{\frac{1}{p}}   |\le |  x-y   |^{\frac{1}{p}}.
\end{equation}
We can assume w.l.o.g. that $x<y$ holds. Then since
\begin{equation}
| x^{\frac{1}{p}}   -y^{\frac{1}{p}}   | = |y|^{\frac{1}{p}} \left( 1- \left(  \frac{x}{y}\right)^{\frac{1}{p}}     \right)      
\end{equation}
it suffices to prove that for all $a\in [0,1]$ and $p\ge 1$ we have
\begin{equation}
1-a^{\frac{1}{p}}\le \left(  1-a     \right) ^{\frac{1}{p}}  ,
\end{equation}
which in turn is equivalent to
\begin{equation}
1\le a^{\frac{1}{p}} +\left(  1-a     \right) ^{\frac{1}{p}}.
\end{equation}
Since this clearly holds for $a\in [0,1]$ and $p\ge 1$, we can conclude that (\ref{eq:norm-mon}) holds.
\end{enumerate}
%%%%%%%%%%%%%%%%%%%%%%%%%%%%%%%%%%%%%%%%%%%%%%%%%%%%
%
%                                     THEOREM: APPROXIMATION OF FUNCTIONS ACHIEVABILITY
%
%%%%%%%%%%%%%%%%%%%%%%%%%%%%%%%%%%%%%%%%%%%%%%%%%%%
We are now in a position to state our main theorem on approximation of functions in $ \fr_{\textrm{\textrm{mon}}}$.
To this end, we introduce the notion of total spread of  the inner part of $ f\in  \fr_{\textrm{\textrm{mon}}}$ as
\begin{equation}\label{eq:total-inner-spread}
\bar{\Delta}(f):=\sum_{k=1}^K ( \phi_{\max,k}-\phi_{\min,k}),
\end{equation}
along with the $\max $-spread
\begin{equation}\label{eq:max-inner-spread}
\Delta (f):= \max_{1\le k \le K} ( \phi_{\max,k}-\phi_{\min,k}),
\end{equation}
where
 \begin{equation}\label{eq:phi-def-spread}
 \phi_{\min,k}:= \inf_{s \in \mathcal{S}_k} f_k (s), \quad \phi_{\max,k}:=\sup_{s \in \mathcal{S}_k} f_k( s).
 \end{equation}
 We define the relative spread with power constraint $P$ as
 \begin{equation}\label{eq:relative-spread}
 \Delta (f\| P):= P \cdot \frac{\bar{\Delta} (f)}{\Delta (f)}.
 \end{equation}

\begin{theorem}\label{th:approximation-of-functions}
Let $f\in  \fr_{\textrm{\textrm{mon}}}$, $M\in \NN$, and the power constraint $P \in \reals_{+} $ be given. Let $\Phi$ be an increment majorant of $f$. Then for any $\eps>0$, there exist pre-processing and post-processing operations creating the estimate $ \bar{f}$
such that upon $M$ uses of the channel (\ref{eq:channel-model}) we have
\begin{multline}\label{eq:approximation-error-theorem}
 \mathbb{P} (  |    \bar{f} (s^K)- f(s^K)       |\ge \eps       ) \le  \Gamma_{1,M} (\eps, K,f,\sigma_F)\\
 + \Gamma_{2,M}(\eps, K, f, \sigma_N, \sigma_F, P),
 \end{multline}
 for all $s^K\in \Ss_1\times \ldots \times \Ss_K$, where
 \begin{itemize}
 \item $\Gamma_{1,M} (\eps,K,f,\sigma_F)\ $ is given by
 \begin{multline}\label{eq:approximation-error-theorem-term1}
 \Gamma_{1,M} (\eps, K,f,\sigma_F)
 \\ = 2 \exp \left( -\frac{M \eta^2}{2 \Delta(f) \sigma_F^2 \eta + 8 \Delta(f)^2 K \sigma_F^4} \right),
  \end{multline}
  where $\eta = \Phi^{-1}(\varepsilon)/2$.
  \item $\Gamma_{2,M}(\eps,K, f, \sigma_N, \sigma_F, P) $ is given by
  \begin{align}\label{eq:approximation-error-theorem-term2}
   \Gamma_{2,M}(\eps,K, f,\sigma_N, \sigma_F, P)=
   2 \exp \left (   -\frac{M \eta^2}{2 L \eta + 4  L^2}   \right ),
   \end{align}
where $L= 3\sigma_F^2 \bar{\Delta}(f)+\frac{4\sigma_N\sigma_F\sqrt{\Delta(f)\bar{\Delta}(f)}  }{\sqrt{P}}+\frac{2\sigma_{N}^2 \Delta(f)}{P}$ and $\eta = \Phi^{-1}(\varepsilon)/2$.

    \end{itemize}
\end{theorem}
\begin{remark}
\label{remark:communication-cost}
This theorem implies a bound on the communication cost $M(f,\varepsilon,\delta)$. Indeed, we can upper bound (\ref{eq:approximation-error-theorem}) as
\begin{multline*}
 \mathbb{P} (  |    \bar{f} (s^K)- f(s^K)       |\ge \eps       ) \\
 \le 
 2\max\left(
   \Gamma_{1,M} (\eps, K,f,\sigma_F),
   \Gamma_{2,M}(\eps, K, f, \sigma_N, \sigma_F, P)
 \right)
\end{multline*}
and solve the expression for $M$ to obtain
\begin{align}\label{eq:commcost}
M(f,\varepsilon,\delta) \leq \frac{\log 4 - \log \delta}{\Phi^{-1}(\varepsilon)^2} \max\left(\gamma_1(f,\varepsilon,\delta), \gamma_2(f,\varepsilon,\delta)\right),
\end{align}
where
\begin{align}
\gamma_1(f,\varepsilon,\delta)
&=
4 \Delta(f) \sigma_F^2\Phi^{-1}(\varepsilon) + 32 \Delta(f)^2K\sigma_F^4
\\
\gamma_2(f,\varepsilon,\delta)
&=
4L\Phi^{-1}(\varepsilon)+16L^2.
\end{align}
\end{remark}
\begin{example}\label{example:sumfunction}
Consider the sum function
\[
f: [0,1]^K \rightarrow \reals,~ (s_1, \dots, s_K) \mapsto s_1 + \dots + s_K.
\]
By Definition~\ref{def:Fmon}, $f \in \fr_{\textrm{\textrm{mon}}}$, where $f_1, \dots, f_K, F$ and $\Phi$ are all equal to the identity function. We therefore have $\bar{\Delta}(f) = K$, $\Delta(f) = 1$ and $\eta = \varepsilon/2$. Putting these into (\ref{eq:approximation-error-theorem-term1}) and (\ref{eq:approximation-error-theorem-term2}), we have
\begin{align}
\label{eq:examplesum-Gamma1}
\Gamma_{1,M} (\eps, K,f,\sigma_F)
&=
2 \exp \left( -\frac{M \varepsilon^2}{4 \sigma_F^2 \varepsilon + 32 K \sigma_F^4} \right)
\\
\label{eq:examplesum-Gamma2}
\Gamma_{2,M}(\eps,K, f,\sigma_N, \sigma_F, P)
&=
2 \exp \left (   -\frac{M \varepsilon^2}{4 L \varepsilon + 16  L^2}   \right ),
\end{align}
where
\begin{align}\label{eq:sumfunction-L}
L= 3\sigma_F^2 K+\frac{4\sigma_N\sigma_F  \sqrt{K}}{\sqrt{P}}+\frac{2\sigma_{N}^2}{P}.
\end{align}
We can also view this in term of the communication cost; i.e., for this $f$, (\ref{eq:commcost}) holds with $\Phi^{-1}(\varepsilon) = \varepsilon$ and
\begin{align}
\gamma_1(f,\varepsilon,\delta)
&=
4 \sigma_F^2\varepsilon + 32 K\sigma_F^4
\\
\gamma_2(f,\varepsilon,\delta)
&=
4L\varepsilon+16L^2,
\end{align}
where $L$ is given in (\ref{eq:sumfunction-L}).

Therefore, if we want to achieve a bounded approximation error in the case $K \rightarrow \infty$, we have to let $M$ grow proportionally with $K^2$.
\end{example}
\begin{example}\label{example:arithmeticavgfunction}
Consider the arithmetic average function
\[
f: [0,1]^K \rightarrow \reals,~ (s_1, \dots, s_K) \mapsto \frac{s_1 + \dots + s_K}{K}.
\]
The situation is almost the same as in Example~\ref{example:sumfunction}, except that $F:s \mapsto s/K$ and $\Phi = F$. Therefore, $\Delta(f)$ and $\bar{\Delta}(f)$ are as in Example~\ref{example:sumfunction}, while $\eta = K \varepsilon / 2$. Substituting these into (\ref{eq:approximation-error-theorem-term1}) and (\ref{eq:approximation-error-theorem-term2}) yields
\begin{align}
\Gamma_{1,M} (\eps, K,f,\sigma_F)
&=
2 \exp \left( -\frac{M K \varepsilon^2}{4 \sigma_F^2 \varepsilon + 32 \sigma_F^4} \right)
\\
\Gamma_{2,M}(\eps,K, f,\sigma_N, \sigma_F, P)
&=
2 \exp \left (   -\frac{M \varepsilon^2}{4 L \varepsilon + 16  L^2}   \right ),
\end{align}
where
\begin{align}\label{eq:avgfunction-L}
L= 3\sigma_F^2+\frac{4\sigma_N\sigma_F}{\sqrt{PK}}+\frac{2\sigma_{N}^2}{PK}.
\end{align}
Therefore, we can achieve a bounded approximation error in the case $K \rightarrow \infty$ without having to let $M$ grow with $K$.
\end{example}
\begin{example}
Consider the $2$-norm of vectors in a hypercube
\[
f: [-1,1]^K \rightarrow \reals, (s_1, \dots, s_K) \mapsto \sqrt{s_1^2 + \dots + s_K^2}.
\]
We have $f_1 = \dots = f_K: s \mapsto s^2$ and $F = \Phi: s \mapsto \sqrt{s}$. Therefore, $\Delta(f) = 1$, $\bar{\Delta}(f) = K$ and $\eta = \varepsilon^2/2$. Substituting these into (\ref{eq:approximation-error-theorem-term1}) and (\ref{eq:approximation-error-theorem-term2}), we get
\begin{align}
\Gamma_{1,M} (\eps, K,f,\sigma_F)
&=
2 \exp \left( -\frac{M \varepsilon^4}{4 \sigma_F^2 \varepsilon^2 + 32 K \sigma_F^4} \right)
\\
\Gamma_{2,M}(\eps,K, f,\sigma_N, \sigma_F, P)
&=
2 \exp \left (   -\frac{M \varepsilon^4}{4 L \varepsilon^2 + 16  L^2}   \right ),
\end{align}
where
$L= 3\sigma_F^2 K+\frac{4\sigma_N\sigma_F  \sqrt{K}}{\sqrt{P}}+\frac{2\sigma_{N}^2}{P}$.
\end{example}

%%%%%%%%%%%%%%%%%%%%%%%%%%%%%%%%%%%%%%%%%%%%%%%%%%%%%%%%%%%%%%%%%%%
%
%                                                     Distributed Function Approximation in Machine Learning
%
%%%%%%%%%%%%%%%%%%%%%%%%%%%%%%%%%%%%%%%%%%%%%%%%%%%%%%%%%%%%%%%%%%%
%
\section{Distributed Function Approximation in Machine Learning}
\label{sec:ml}
{
\newcommand{\mlInputAlphabet}{{\mathcal{X}}}
\newcommand{\mlLabelAlphabet}{{\mathcal{Y}}}
\newcommand{\mlInputAlphabetElement}{{x}}
\newcommand{\mlLabelAlphabetElement}{{y}}
\newcommand{\mlEstimator}{{f}}
\newcommand{\mlLoss}{{L}}
\newcommand{\mlDistribution}{{\mathcal{P}}}
\newcommand{\Expectation}{{\mathbb{E}}}
\newcommand{\mlInputRV}{{X}}
\newcommand{\mlOutputRV}{{Y}}
\newcommand{\mlRisk}[2]{{\mathcal{R}_{{#1},{#2}}}}
\newcommand{\mlEstimatorOutput}{{t}}
\newcommand{\mlKernel}{{\kappa}}
\newcommand{\numUsers}{{K}}
\newcommand{\indexUsers}{{k}}
\newcommand{\RKHS}{{\mathcal{H}}}
\newcommand{\mlTrainingSampleNum}{{N}}
\newcommand{\mlIndexTrainingSample}{{n}}
\newcommand{\mlEstimatorCoefficient}{{\alpha}}
\newcommand{\partfuncminvalue}[1]{{\phi_{\min,{#1}}}}
\newcommand{\partfuncmaxvalue}[1]{{\phi_{\max,{#1}}}}
\newcommand{\errorconstone}{{\varepsilon}}
\newcommand{\errorconsttwo}{{\delta}}
\newcommand{\Probability}{{\mathbb{P}}}
\newcommand{\absolute}[1]{{\lvert {#1} \rvert}}
\newcommand{\mlLossLipschitz}{{B}}

In this section, we discuss how the methods described in this paper can be used to compute the estimators of support vector machines (SVM) in a distributed fashion. First, we briefly sketch the setting as in~\cite{steinwart}. We consider an input alphabet $\mlInputAlphabet$, a label alphabet $\mlLabelAlphabet \subseteq \reals$ and a probability distribution $\mlDistribution$ on $\mlInputAlphabet \times \mlLabelAlphabet$ which is in general unknown. A statistical inference problem is characterized by the input alphabet, the label alphabet and a loss function $\mlLoss: \mlInputAlphabet \times \mlLabelAlphabet \times \reals \rightarrow [0, \infty)$. The objective is, given training samples drawn i.i.d. from $\mlDistribution$, to find an estimator function $\mlEstimator: \mlInputAlphabet \rightarrow \reals$ such that the risk $\mlRisk{\mlLoss}{\mlDistribution} := \Expectation_\mlDistribution \mlLoss(\mlInputRV, \mlOutputRV, \mlEstimator(\mlInputRV))$ is as small as possible. In order for the risk to exist, we must impose suitable measurability conditions on $\mlLoss$ and $\mlEstimator$. In this paper, we deal with Lipschitz-continuous losses. We say that the loss $\mlLoss$ is $\mlLossLipschitz$-Lipschitz-continuous if $\mlLoss(\mlInputAlphabetElement, \mlLabelAlphabetElement, \cdot)$ is Lipschitz-continuous for all $\mlInputAlphabetElement \in \mlInputAlphabet$ and $\mlLabelAlphabetElement \in \mlLabelAlphabet$ with a Lipschitz constant uniformly bounded by $\mlLossLipschitz$. Lipschitz-continuity of a loss function is a property that is also often needed in other contexts. Fortunately, many loss functions of practical interest possess this property. For instance, the absolute distance loss, the logistic loss, the Huber loss and the $\varepsilon$-insensitive loss, all of which are commonly used in regression problems~\cite[Section 2.4]{steinwart}, are Lipschitz-continuous. Even in scenarios in which the naturally arising loss is not Lipschitz-continuous, for the purpose of designing the machine learning model, it is often replaced with a Lipschitz-continuous alternative. For instance, in binary classification, we have $\mlLabelAlphabet = \{-1,1\}$ and the loss function is given by
\[
(\mlInputAlphabetElement,\mlLabelAlphabetElement,\mlEstimatorOutput) \mapsto
\begin{cases}
  0, &\sign(\mlLabelAlphabetElement) = \sign(\mlEstimatorOutput) \\
  1, &\text{otherwise.}
\end{cases}
\]
This loss is not even continuous, which makes it hard to deal with. So for the purpose of designing the machine learning model, it is commonly replaced with the Lipschitz-continuous hinge loss or logistic loss~\cite[Section 2.3]{steinwart}.

Here, we consider the case in which the inputs are $\numUsers$-tuples and the SVM can be trained in a centralized fashion. The actual predictions, however, are performed in a distributed setting; i.e., there are $\numUsers$ users each of which observes only one component of the input. The objective is to make an estimate of the label available at the receiver while using as little communication resources as possible.

To this end, we consider the case of additive models which is described in~\cite[Section 3.1]{christmann2012consistency}. We have $\mlInputAlphabet = \mlInputAlphabet_1 \times \dots \times \mlInputAlphabet_\numUsers$ and a kernel $\mlKernel_\indexUsers: \mlInputAlphabet_\indexUsers \times \mlInputAlphabet_\indexUsers \rightarrow \reals$ with an associated reproducing kernel Hilbert space $\RKHS_\indexUsers$ of functions mapping from $\mlInputAlphabet_\indexUsers$ to $\reals$ for each $\indexUsers \in \{1,\dots,\numUsers\}$. Then by~\cite[Theorem 2]{christmann2012consistency}
\begin{multline}\label{eq:addkernel}
\mlKernel: \mlInputAlphabet \times \mlInputAlphabet \rightarrow \reals,~
((\mlInputAlphabetElement_1,\dots,\mlInputAlphabetElement_\numUsers), (\mlInputAlphabetElement'_1,\dots,\mlInputAlphabetElement'_\numUsers))
\mapsto \\
\mlKernel_1(\mlInputAlphabetElement_1, \mlInputAlphabetElement'_1) + \dots + \mlKernel_\numUsers(\mlInputAlphabetElement_\numUsers, \mlInputAlphabetElement'_\numUsers)
\end{multline}
is a kernel and the associated reproducing kernel Hilbert space is
\begin{align}\label{eq:addkernel-rkhs}
\RKHS := \{\mlEstimator_1 + \dots + \mlEstimator_\numUsers: \mlEstimator_1 \in \RKHS_1, \dots, \mlEstimator_\numUsers \in \RKHS_\numUsers\}.
\end{align}
So this model is appropriate whenever the function to be approximated is expected to have an additive structure. We know~\cite[Theorem 5.5]{steinwart} that an SVM estimator has the form
\begin{align}\label{eq:svmestimator}
\mlEstimator(\mlInputAlphabetElement) = \sum_{\mlIndexTrainingSample=1}^\mlTrainingSampleNum \mlEstimatorCoefficient_\mlIndexTrainingSample \mlKernel(\mlInputAlphabetElement, \mlInputAlphabetElement^\mlIndexTrainingSample),
\end{align}
where $\mlEstimatorCoefficient_1, \dots \mlEstimatorCoefficient_\mlTrainingSampleNum \in \reals$ and $\mlInputAlphabetElement^1, \dots \mlInputAlphabetElement^\mlTrainingSampleNum \in \mlInputAlphabet$. In our additive model, this is
\begin{align}\label{eq:addestimator}
\mlEstimator(\mlInputAlphabetElement_1, \dots, \mlInputAlphabetElement_\indexUsers)
=
\sum_{\indexUsers=1}^\numUsers \mlEstimator_\indexUsers(\mlInputAlphabetElement_\indexUsers),
\end{align}
where for each $\indexUsers$,
\begin{align}\label{eq:addestimator-detail}
\mlEstimator_\indexUsers(\mlInputAlphabetElement_\indexUsers) =
\sum_{\mlIndexTrainingSample=1}^\mlTrainingSampleNum \mlEstimatorCoefficient_\mlIndexTrainingSample \mlKernel_\indexUsers(\mlInputAlphabetElement_\indexUsers, \mlInputAlphabetElement_\indexUsers^\mlIndexTrainingSample).
\end{align}

We can now state a result for the distributed approximation of the estimator of such an additive model as an immediate corollary to Theorem~\ref{th:approximation-of-functions}.

\begin{cor}
Consider an additive machine learning model, i.e., we have an estimator of the form (\ref{eq:addestimator}), and assume that $\mlLoss$ is a $\mlLossLipschitz$-Lipschitz-continuous loss. Suppose further that all the $\mlEstimator_\numUsers$ have bounded range such that the quantities $\bar{\Delta}(f)$ and $\Delta(f)$ as defined in (\ref{eq:total-inner-spread}) and (\ref{eq:max-inner-spread}) exist and are finite. Let $\errorconstone, \errorconsttwo > 0$ and $M \geq M(\mlEstimator,\errorconstone, \errorconsttwo)$ as defined in (\ref{eq:commcost}), where $\Phi^{-1}(\varepsilon) = \varepsilon$. Then, given any $\mlInputAlphabetElement^\numUsers = (\mlInputAlphabetElement_1, \dots, \mlInputAlphabetElement_\numUsers)$ at the transmitters and any $\mlLabelAlphabetElement \in \mlLabelAlphabet$, through $M$ uses of the channel (\ref{eq:channel-model}), the receiver can obtain an estimate $\bar{\mlEstimator}$ of $\mlEstimator(\mlInputAlphabetElement^\numUsers)$ satisfying
\begin{align}
\label{eq:mlapplication-cor}
\Probability(\absolute{\mlLoss(\mlInputAlphabetElement^\numUsers,\mlLabelAlphabetElement,\bar{\mlEstimator}) - \mlLoss(\mlInputAlphabetElement^\numUsers,\mlLabelAlphabetElement,\mlEstimator(\mlInputAlphabetElement^\numUsers))} \geq \mlLossLipschitz \errorconstone)
\leq \errorconsttwo.
\end{align}
\end{cor}

\begin{proof}
The Lipschitz continuity of $\mlLoss$ yields
\begin{multline*}
\Probability(\absolute{\mlLoss(\mlInputAlphabetElement^\numUsers,\mlLabelAlphabetElement,\bar{\mlEstimator}) - \mlLoss(\mlInputAlphabetElement^\numUsers,\mlLabelAlphabetElement,\mlEstimator(\mlInputAlphabetElement^\numUsers))} \geq \mlLossLipschitz \errorconstone)
\\
\leq
\Probability(\absolute{\bar{\mlEstimator} - \mlEstimator(\mlInputAlphabetElement^\numUsers)} \geq \errorconstone),
\end{multline*}
from which (\ref{eq:mlapplication-cor}) follows by Remark~\ref{remark:communication-cost}.
\end{proof}

We conclude this section with a brief discussion of the feasibility of the condition that $\mlEstimator_1, \dots, \mlEstimator_\numUsers$ have bounded ranges in the case of the additive SVM model discussed above. The coefficients $\mlEstimatorCoefficient_1, \dots, \mlEstimatorCoefficient_\mlTrainingSampleNum$ are a result of the training step and can therefore be considered constant, so all we need is for the ranges of $\mlKernel_1, \dots, \mlKernel_\numUsers$ to be bounded. This heavily depends on $\mlInputAlphabet_1, \dots, \mlInputAlphabet_\numUsers$ and the choices of the kernels, but we remark that the boundedness criterion is satisfied in many cases of interest. The range of Gaussian kernels is always a subset of $(0,1]$, and while other frequent choices such as exponential, polynomial and linear kernels can have arbitrarily large ranges, they are nonetheless continuous which means that as long as the input alphabets are compact topological spaces (e.g. closed hyperrectangles or balls), the ranges are also compact, and therefore bounded.
}
%%%%%%%%%%%%%%%%%%%%%%%%%%%%%%%%%%%%%%%%%%%%%%%%%%%%%%%%%%%%%%%%%%%
%
%                                                     Application to Max-consensus problem
%
%%%%%%%%%%%%%%%%%%%%%%%%%%%%%%%%%%%%%%%%%%%%%%%%%%%%%%%%%%%%%%%%%%%
%
\section{Application to the Max-Consensus Problem}
\label{sec:max-consensus}
{
\newcommand{\coordinator}{{C}}
\newcommand{\agent}[1]{{A_{#1}}}
\newcommand{\agentSet}{\ensuremath{\mathcal{A}}}
\newcommand{\edgeSet}{\ensuremath{\mathcal{E}}}
\newcommand{\numAgents}{{K}}
\newcommand{\agentIndex}{{k}}
\newcommand{\noise}{{N}}
\newcommand{\agentNoise}[1]{\ensuremath{N_{#1}}}
\newcommand{\channelFading}[1]{\ensuremath{h_{#1}}}
\newcommand{\agentTx}[1]{{\alpha_{#1}}}
\newcommand{\agentRx}[1]{\ensuremath{\Gamma_{#1}}}
\newcommand{\coordinatorRx}{{\gamma}}
\newcommand{\coordinatorTx}{\ensuremath{\beta}}
\newcommand{\agentTxDomain}{{[0,1]}}
\newcommand{\lexLess}{{<}}
\newcommand{\lexGreater}{{>}}
\newcommand{\lexLessCompatible}{{\leq}}
\newcommand{\lexGreaterCompatible}{{\geq}}
\newcommand{\agentInputSequence}[1]{{S_{#1}}}
\newcommand{\coordinatorOutputEstimate}{{S}}
\newcommand{\naturals}{{\mathbb{N}}}
\newcommand{\infiniteBinarySequences}{{\{0,1\}^\infty}}
\newcommand{\finiteBinarySequences}{{\{0,1\}^{<\infty}}}
\newcommand{\binarySequences}{{\{0,1\}^{\leq\infty}}}
\newcommand{\generalBinarySequence}{{S}}
\newcommand{\outputCondition}{{\varphi}}
\newcommand{\outputConditionFreeVariable}{{x}}
\newcommand{\maximumRemainingAgents}{{m}}
\newcommand{\maximumRemainingAgentsSet}{{\mathcal{M}}}
\newcommand{\stepIndex}{{t}}
\newcommand{\emptySequence}{{\emptyset}}
\newcommand{\activeAgents}{{\mathcal{A}}}
\newcommand{\protestingAgents}{{\mathcal{P}}}
\newcommand{\raisingAgents}{{\mathcal{R}}}
\newcommand{\append}{{{}^\frown}}
\newcommand{\coordinatorTerminationCount}{{T}}
\newcommand{\cardinality}[1]{{\lvert{#1}\rvert}}
\newcommand{\Probability}{{\mathbb{P}}}
\newcommand{\goodStateEvent}{{\mathbb{G}}}
\newcommand{\badStateEvent}{{\tilde{\mathbb{G}}}}
\newcommand{\goodTermination}{{\mathbb{T}}}
\newcommand{\badTermination}{{\tilde{\mathbb{T}}}}
\newcommand{\probabilitySpace}{{\Omega}}
\newcommand{\descriptionLength}{{d}}
\newcommand{\currentLength}{{\ell}}
\newcommand{\correctionEvent}{{\mathcal{C}}}
\newcommand{\randomWalkLike}[1]{{R(#1)}}
\newcommand{\randomWalkLikeModified}[1]{{R'(#1)}}
\newcommand{\actualRandomWalk}[1]{{\bar{R}(#1)}}
\newcommand{\badDigits}{{b}}
\newcommand{\goodDigits}{{g}}
\newcommand{\terminationThreshold}{{\tau}}
\newcommand{\agentInputSet}{{\mathcal{S}}}
\newcommand{\networkGraph}{{\mathcal{G}}}
\newcommand{\numCoordinators}{{c}}
\newcommand{\coordinatorIndex}{{\ell}}
\newcommand{\indicator}[1]{{\mathbf{1}_{{#1}}}}
\newcommand{\compatible}{||}
\newcommand{\generalindex}{{k}}
\newcommand{\generalbit}{{b}}
\newcommand{\residualProbability}{\varepsilon}
\newcommand{\quantizationPrecision}{{p}}
\newcommand{\absolute}[1]{{\lvert {#1} \rvert}}
\newcommand{\terminationRelation}{{R}}
\newcommand{\errorconsttwo}{{\gamma}}

In this section, we consider the problem of achieving max-consensus in a network of agents, which is relevant in many practical applications such as task assignment, leader election, rendezvous, clock synchronization, spectrum sensing, distributed decision making and formation control. In our previous work~\cite{agrawal2019scalable}, we propose the \emph{ScalableMax} scheme for achieving max-consensus in large star-shaped networks and an extension to not necessarily star-shaped but nonetheless highly connected networks. A notable restriction in~\cite{agrawal2019scalable} is the assumption that the fading coefficients of the wireless multiple-access channel considered are all deterministically equal to $1$. In this section, we show how this restriction can be lifted in light of Theorem~\ref{th:approximation-of-functions} of the present work so as to accommodate a fast fading channel.

In the max-consensus problem, we consider $\numAgents$ agents, all of which hold an input value from a totally ordered set (e.g., a real interval). We say that max-consensus has been achieved if all agents hold an output estimate which is equal to the maximum of the inputs. The objective of the max-consensus problem is to achieve max-consensus with as little usage of communication resources as possible.

The \emph{ScalableMax} scheme solves a somewhat simpler problem, called \emph{weak $\maximumRemainingAgents$-max-consensus}, where  $\maximumRemainingAgents$ is a designable parameter that can depend on the noise, but is usually independent of the size of the network. It is argued that once weak $\maximumRemainingAgents$-max-consensus has been achieved, state-of-the-art methods allow for reaching max-consensus in a number of channel uses linear in $\maximumRemainingAgents$. The higher $\maximumRemainingAgents$ is, the more resilient the scheme is against noise, but the more channel resources are needed for this follow-up scheme.

The communication model assumed in~\cite{agrawal2019scalable} is that the coordinator can multicast digital information noiselessly to all the agents simultaneously (which can in a practical application be achieved with state-of-the-art coding techniques even if the actually available channel is noisy) and that the agents can simultaneously transmit analog messages to the coordinator through a multiple-access channel $\coordinatorRx = \sum\nolimits_{k=1}^\numAgents \agentTx{k} + \noise$, where $\agentTx{1}, \dots, \agentTx{\numAgents}$ are the channel inputs, $\coordinatorRx$ is the channel output and $\noise$ is arbitrary noise, the tail probabilities of which can be bounded. One communication \emph{step} is a combination of one noiseless multicast transmission from the coordinator to the agents and three uses of the noisy multiple-access channel from the agents to the coordinator. \cite{agrawal2019scalable} also defines the \emph{maximum description length} $\descriptionLength$, which depends on the agents' inputs, but is, e.g., logarithmic in the number of agents in case their inputs are uniformly randomly distributed over some real interval. We have the following result about achieving weak $\maximumRemainingAgents$-max-consensus in this system:

\begin{theorem}[\cite{agrawal2019scalable}]
\label{theorem:maxconsensus}
Suppose that $\maximumRemainingAgents$ is even. Then the probability that the ScalableMax scheme reaches weak $\maximumRemainingAgents$-max-consensus within $\descriptionLength+1$ steps is at least $\Probability(\absolute{\noise} \leq \maximumRemainingAgents/4)^{3(\descriptionLength+1)}$.
\end{theorem}

As a direct corollary of Theorem~\ref{th:approximation-of-functions}, we can now replace one use of the idealized channel assumed in Theorem~\ref{theorem:maxconsensus} with $M$ uses of the channel (\ref{eq:channel-model}), and obtain the following result.
\begin{cor}
Fix some integer $M>0$ and some even integer $\maximumRemainingAgents > 0$. Define $\errorconsttwo$ to be the right hand side of (\ref{eq:approximation-error-theorem}), where $\Gamma_{1,M}$ and $\Gamma_{2,M}$ are defined as in (\ref{eq:examplesum-Gamma1}) and (\ref{eq:examplesum-Gamma2}) with $\varepsilon := \maximumRemainingAgents/4$. Then weak $\maximumRemainingAgents$-max-consensus can be reached with at most $3M(\descriptionLength+1)$ uses of the channel (\ref{eq:channel-model}) and $M(\descriptionLength+1)$ digital multicast transmissions from the coordinator to all agents with probability at least $(1-\errorconsttwo)^{3(\descriptionLength+1)}$.
\end{cor}

We remark that an issue of scalability remains as an open problem. In order to achieve a constant tail probability for the noise as $\numAgents$ grows, it is shown in Example~\ref{example:sumfunction} that the number of channel uses has to grow linearly with $\numAgents^2$, leading to a scaling in the number of overall channel uses which is less favorable than the logarithmic growth of the number of applications of the distributed approximation scheme suggests.
}
%%%%%%%%%%%%%%%%%%%%%%%%%%%%%%%%%%%%%%%%%%%%%%%%%%%%%%%%%%%%%%%%%%%
%
%	                                                    PROOF OF THEOREM 1
%
%%%%%%%%%%%%%%%%%%%%%%%%%%%%%%%%%%%%%%%%%%%%%%%%%%%%%%%%%%%%%%%%%%%
%
\section{ Proof of Theorem \ref{th:approximation-of-functions}}
The proof of Theorem \ref{th:approximation-of-functions} is organized as follows: In Sections \ref{sec:pre-proc} and \ref{sec:post-proc}, we give our pre- and post-processing operations explicitly. They are, up to slight
modifications, already given in \cite{kiril}.
The basic structure of the approximation error event is discussed in Section \ref{sec:error-event}. Basically, we split the error event into two parts. One part is determined by the fading process and the other part is
caused by the effective noise composed of fading process, additive channel noise, and the randomness introduced during pre-processing.
It turns out that all involved random variables are sub-exponential random variables. Therefore, in Section \ref{sec:sub-exp-sub-gauss} we recall the necessary definitions
 \showto{arxiv}{and basic properties} of this type of random variables and derive bounds on the sub-exponential semi-norm of the random variables appearing in our proof. Moreover, we recall Bernstein's inequality for sub-exponential random variables \showto{arxiv}{and the principle of rotational invariance of sub-gaussian random variables}
 which will be used to derive 
 bounds on the probability of the approximation error event in Section \ref{sec:performance-bounds}.
%%%%%%%%%%%%%%%%%%%%%%%%%%%%%%%%%%%%%%%%%%%%%%%%%%%%%%%%%%%%%%%%%%%
%
%			Pre-Processing
%
%%%%%%%%%%%%%%%%%%%%%%%%%%%%%%%%%%%%%%%%%%%%%%%%%%%%%%%%%%%%%%%%%%%
\subsection{Pre-Processing}\label{sec:pre-proc}
In the pre-processing step we encode the function values $f_k(s_k)$, $k=1,\ldots ,K$ as transmit power:  
\begin{equation}
X_k(m):= \sqrt{g_k(f_k (s_k))} U_k (m),
1\leq m\leq M
\end{equation}
with $g_k: [ \phi_{\min,k}, \phi_{\max,k}] \to [ 0, P  ] $
such that
\begin{equation}\label{eq:pre-processing}
  g_k(t):= \frac{P}{\Delta(f)} (t-\phi_{\min,k}),
  \end{equation}
 where $\Delta(f)$ is given in (\ref{eq:max-inner-spread}) and   $\phi_{\min,k}$ is defined in (\ref{eq:phi-def-spread}).\\
 $U_k (m)$, $k=1, \ldots, K$, $m=1, \ldots ,M$ are i.i.d. with the uniform distribution on $\{-1,+1\}$. We assume the random variables $U_k (m)$, $k=1,\ldots,K$, $m=1\ldots,M$, are independent of
 $H_k(m)$, $k=1,\ldots, K$, $m=1,\ldots,M$, and $N(m)$, $m=1,\ldots,M$.
%%%%%%%%%%%%%%%%%%%%%%%%%%%%%%%%%%%%%%%%%%%%%%%%%%%%%%%%%%%%%%%%%%%
%
%			Post-Processing
%
%%%%%%%%%%%%%%%%%%%%%%%%%%%%%%%%%%%%%%%%%%%%%%%%%%%%%%%%%%%%%%%%%%%
\subsection{Post-Processing}\label{sec:post-proc}
 The post-processing is based on receive energy which has the form
\begin{equation}
\label{eq:energy}    
\tilde{Y}_{s^K}  = \sum_{m=1}^{M} | Y (m)|^2    = \sum_{k=1}^K  g_k(f_k (s_k)) \| H_k \|_2^2
+ \bar{N}_{s^{K}},
\end{equation}
where $H_k=(H_k(1), \ldots , H_k(M))$ is the vector consisting of fading coefficients, and
 $\bar{N}_{s^{K}}= \sum_{m=1}^M \bar{N}_{s^K}(m)$. The random variables $\bar{N}_{s^K}(m) $, $ m=1, \ldots , M$, are given by
\begin{multline}\label{eq:N-s-k}
\bar{N}_{s^K} (m) :=  \\
\begin{aligned}
  &\hphantom{+} \sum_{\substack{k ,l=1, \\ k\neq l}}^K \sqrt{g_k(f_k( s_k)) g_l(f_l (s_l))} H_k(m)\overline{H_l (m)} \\
  &\hphantom{+\sum_{\substack{k ,l=1, \\ k\neq l}}^K} \times  U_{k}(m)U_{l}(m)  \\ 
  & + 2 \Real \left( \overline{N(m)} \sum_{k=1}^K \sqrt{g_k(f_k (s_k))} H_k(m)U_{k}(m) \right ) \\
  & + | N(m) |^2,
\end{aligned}
\end{multline}

and are independent for any $s^K =(s_1, \ldots, s_K)\in \Ss_1 \times \ldots \times \Ss_k$.
The receiver applies to $\tilde{Y}_{s^K}$ in (\ref{eq:energy}) the following recovery operations:
\begin{enumerate}
\item  A function $\bar{g}: \mathbb{R}\to \mathbb{R}$
\[
\bar{g}(t) := \frac{\Delta(f)}{2 \cdot M \cdot  P}t + \sum_{k=1}^K\phi_{\min,k},
\]
resulting in
\begin{IEEEeqnarray}{rCl}\label{eq:est-1}
\bar{h}(s^K)&: = & \bar{g}(\tilde{Y}_{s^K}) \nonumber \\
& =&   \sum_{k=1}^K f'_k(s_k) \frac{\| H_k \|_2^2}{2 M} + \frac{\alpha}{M} \bar{N}_{s^K} \nonumber\\
& & + \sum_{k=1}^K \phi_{\min,k}
\end{IEEEeqnarray}
where $f'_k(s_k):= f_k(s_k)- \phi_{\min,k}$, and  $\alpha:=\frac{\Delta (f)}{2 P}$.
\item  Moreover, since from (\ref{eq:N-s-k}) it follows that for all $m=1,\ldots , M$
 \begin{equation}
 \mathbb{E}(\bar{N}_{s^K}(m)  )= \E (( N^r(m)  )^2) + \E (( N^i(m)  )^2),
 \end{equation}
 which is independent of $s^K\in \mathcal{S}_1 \times \ldots \times \mathcal{S}_K$,
 the receiver  can add
 \[- \frac{\alpha}{M}\mathbb{E}(\bar{N}_{s^K})= - \frac{\alpha}{M}\sum_{m=1}^M \E (\bar{N}_{s^K}(m) )\]
 to (\ref{eq:est-1}) and obtains
\begin{IEEEeqnarray}{rCl}\label{eq:linear-pert}
\tilde{h}(s^K)& := &  \sum_{k=1}^K f'_k(s_k)\frac{\| H_k \|_2^2}{M} % \nonumber \\
  + \frac{\alpha}{M}(\bar{N}_{s^K}- \mathbb{E}(\bar{N}_{s^K})) \nonumber \\
&& + \sum_{k=1}^K \phi_{\min,k}.  
\end{IEEEeqnarray}
\item Finally, the receiver applies the outer function F in order to obtain an estimate of the function $f$:
\begin{equation}\label{eq:noisy-estimate}
\bar{f}(s^K):= F( \tilde{h}(s^K)),
\end{equation}
where $\tilde{h}(s^K)$ is given in (\ref{eq:linear-pert}).
\end{enumerate}
%%%%%%%%%%%%%%%%%%%%%%%%%%%%%%%%%%%%%%%%%%%%%%%%%%%%%%%%%%%%%%%%%%
%
%                                  The Error Event
%
%%%%%%%%%%%%%%%%%%%%%%%%%%%%%%%%%%%%%%%%%%%%%%%%%%%%%%%%%%%%%%%%%%%
\subsection{The Error Event}\label{sec:error-event}
For a given $\eps >0$  and $s^K \in \Ss_1 \times \ldots \times \Ss_K$ we are interested in bounding the probability of the event
\begin{equation}\label{eq:deviation-event}
\left \{  |   \bar{f} (s^K)- f(s^K)| \ge  \eps \right  \}.
\end{equation}
We will now use the assumption that  $f\in \fr_{\textrm{\textrm{mon}}}$ to simplify the deviation event in (\ref{eq:deviation-event}). Since $F$ has the property (\ref{eq:monotone-domination})
we obtain for a function  $\Phi : [0, \infty) \to [0, \infty)$  with $\Phi(0)=0$
\begin{align}
\nonumber
 |   \bar{f} (s^K)- f(s^K)|
 &=
\left | F( \tilde{h}(s^K))- F\left(\sum_{k=1}^Kf_k(s_k)\right) \right | 
\\
\nonumber
&\le  \Phi \left (\left | \tilde{h}(s^K)-   \sum_{k=1}^Kf_k(s_k) \right |\right ) 
\\
\label{eq:error-event-1}
&\begin{aligned}
= \Phi  \Bigg ( \Bigg |   \sum_{k=1}^K &f'_k(s_k) \left(\frac{\|H_k \|_2^2}{2 M}-1\right)
\\
&+ \frac{\alpha}{M}(\bar{N}_{s^K}- \mathbb{E}(\bar{N}_{s^K})) \Bigg |   \Bigg ),
\end{aligned}
\end{align}
 where we have used (\ref{eq:linear-pert}) with $f'_k(s_k):= f_k(s_k)- \phi_{\min,k}$ and $\alpha={\Delta (f)}/{2 P}$.
 This has the following consequence: For any $\eps> 0$ and $s^K \in \Ss_1\times \ldots \times \Ss_K$ we have that
 \begin{equation}
  |   \bar{f} (s^K)- f(s^K)|\ge \eps
   \end{equation}
   implies
   \begin{multline}\label{eq:error-event-2}
 \left |   \sum_{k=1}^K f'_k(s_k) \left(\frac{\|H_k \|_2^2}{2 M}-1\right) + \frac{\alpha}{M}(\bar{N}_{s^K}- \mathbb{E}(\bar{N}_{s^K})) \right | \\ \ge \Phi^{-1}(\eps).
  \end{multline}
 Using (\ref{eq:error-event-2}) and the union bound we can bound the approximation error probability for $s^K\in \Ss_1\times \ldots \times \Ss_K$ as
 \begin{multline}\label{eq:error-event-3}
 \mathbb{P} (  |    \bar{f} (s^K)- f(s^K)       |\ge \eps      )
 \\
 \begin{aligned}
 \le ~
 &\mathbb{P} \left( \left |     \sum_{k=1}^K f'_k(s_k) \left(\frac{\|H_k \|_2^2}{2 M}-1\right)     \right | \ge \frac{\Phi^{-1}(\eps)}{2}         \right)
\\
&+ \mathbb{P}\left (  \left |    \frac{\alpha}{M}(\bar{N}_{s^K}- \mathbb{E}(\bar{N}_{s^K}))    \right | \ge     \frac{\Phi^{-1}(\eps)}{2}         \right).
 \end{aligned}
\end{multline}
In the next subsection we recall some elementary techniques for bounding the probabilities on the right-hand side of (\ref{eq:error-event-3}).
%%%%%%%%%%%%%%%%%%%%%%%%%%%%%%%%%%%%%%%%%%%%%%%%%%%%%%%%%%%%%%%%%%%%
%
%                     Sub-Exponential Random Variables
%
%%%%%%%%%%%%%%%%%%%%%%%%%%%%%%%%%%%%%%%%%%%%%%%%%%%%%%%%%%%%%%%%%%%%
\subsection{Sub-Exponential Random Variables and Bounds}\label{sec:sub-exp-sub-gauss}
In this subsection we will derive bounds on the sub-exponential semi-norms of the random variables 
\begin{equation}
\bar{N}_{s^K}(m)- \mathbb{E}(\bar{N}_{s^K}(m)),
\end{equation}
where $ \bar{N}_{s^K}(m)$ is given in (\ref{eq:N-s-k}), and the summands of
\begin{equation}
\frac{\|H_k \|_2^2}{2 M}-1,
\end{equation}
i.e. 
\begin{equation}
 (H^r_k (m))^2 -1 \textrm{ and } (H^i_k (m))^2 -1,
\end{equation}
for $m=1,\ldots, M$.\\
This will allow us to apply Bernstein's inequality \cite[Chapter 1]{buldygin} for sub-exponential random variables and will lead to exponentially decreasing error  bounds
in (\ref{eq:error-event-3}).\\
In the following we recall some basic definitions \showto{arxiv}{and results} from \cite[Chapter 1]{buldygin}. For a  random variable $X$ we define\footnote{Note that as with our definition of the sub-gaussian norm, other norms on the space of sub-exponential random variables that appear in the literature are equivalent to $\subexpnorm{\cdot}$ (see, e.g.,~\cite{buldygin}). The particular definition we choose here matters, however, because we want to derive results in which no unspecified constants appear.}
\begin{equation}\label{eq:sub-exp-norm-def}
\subexpnorm{X} := \sup\limits_{k \ge 1} \left( \frac{\mathbb{E}(|X|^k)}{k!} \right)^\frac{1}{k}
\end{equation}
If $ \subexpnorm{X}< \infty  $ then $X$ is called a sub-exponential random variable. $\subexpnorm{\cdot}$ defines a semi-norm on the vector space of sub-exponential 
random variables \cite[Remark 1.3.2]{buldygin}. Typical examples of sub-exponential random variables are bounded random variables and random variables with exponential distribution.
\begin{shownto}{arxiv}
We collect some useful  properties of and interrelations between the sub-exponential and sub-gaussian norms in the following lemma.
%%%%%%%%%%%%%%%%%%%%%%%%%%%%%%%%%%%%%%
%
%                                        Lemma: Sub-exponential properties
%
%%%%%%%%%%%%%%%%%%%%%%%%%%%%%%%%%%%%%
\begin{lemma}\label{lemma:sub-exp-properties}
Let $X,Y$ be random variables. Then:
\begin{enumerate}
\item If $X$ is $\mathcal{N}(\mu, \sigma^2)$ then we have
\begin{equation}\label{eq:gaussian-bound}
\subgaussnorm{X} = \sigma.
\end{equation}
\item (Rotation Invariance) If $X_1, \dots, X_M$ are independent, sub-gaussian and centered, we have
\begin{equation}\label{eq:subgauss-rotinv}
\subgaussnorm{\sum\limits_{m=1}^M X_m}^2 \leq \sum\limits_{m=1}^M \subgaussnorm{X_m}^2
\end{equation}
\item If $X$ is a random variable with $| X |\le 1$ with probability $1$ and if $Y$ is independent of $X$ and sub-gaussian then we have
\begin{equation}\label{eq:bounded-12-bound}
\subgaussnorm{X \cdot Y}\le \subgaussnorm{Y}.
\end{equation}
\item If $X$ and $Y$ are sub-gaussian and centered, then $X\cdot Y$ is sub-exponential and
\begin{equation}\label{eq:c-s}
\subexpnorm{X \cdot Y} \le 2 \cdot \subgaussnorm{X} \cdot \subgaussnorm{Y}.
\end{equation}
\item (Centering) If $X$ is sub-exponential and $X \geq 0$ almost surely, then
\begin{equation}\label{eq:centering}
\subexpnorm{X- \mathbb{E}(X)} \leq \subexpnorm{X}.
\end{equation}
\end{enumerate}
\end{lemma}
\begin{proof}
(\ref{eq:gaussian-bound}) follows in a straightforward fashion by calculating the moment generating function of $X$. (\ref{eq:subgauss-rotinv}) is e.g. proven in \cite[Lemma 1.1.7]{buldygin}. (\ref{eq:bounded-12-bound}) follows directly from the definition conditioning on $X$. We show (\ref{eq:c-s}) first for $X=Y$. In this case, we have
\begin{align}
\subexpnorm{X^2} &= \sup_{k \geq 1} \left(\frac{\mathbb{E}X^{2k}}{k!}\right)^{\frac{1}{k}} \leq \sup_{k \geq 1} \left( \frac{2^{k+1}k^k\subgaussnorm{X}^{2k}}{e^k k!} \right)^\frac{1}{k} \nonumber \\
&= 2 \subgaussnorm{X}^2 \sup_{k \geq 1} \left(\frac{2^\frac{1}{k} k}{e (k!)^\frac{1}{k}}\right) \leq 2 \subgaussnorm{X}^2,
\end{align}
where the first inequality is by \cite[Lemma 1.1.4]{buldygin} and the second follows from $2k^k/k! \leq e^k$, which is straightforward to prove for $k\ge 1$ by induction. In the general case, we have
\pagebreak[0]
\begin{align}
\subexpnorm{XY} &= \subgaussnorm{X}\subgaussnorm{Y}\subexpnorm{\frac{XY}{\subgaussnorm{X}\subgaussnorm{Y}}}\\
&\leq \subgaussnorm{X}\subgaussnorm{Y} \subexpnorm{\frac{1}{2} \left(\frac{X}{\subgaussnorm{X}}\right)^2 + \frac{1}{2} \left(\frac{Y}{\subgaussnorm{Y}}\right)^2 } \\
&\leq 2 \subgaussnorm{X}\subgaussnorm{Y},
\end{align}
where the first inequality can be verified in (\ref{eq:sub-exp-norm-def}), considering that $ab \leq a^2/2 + b^2/2$ for all $a,b \in \reals$, and the second inequality follows from the triangle inequality and the special case $X=Y$.

For (\ref{eq:centering}), we assume without loss of generality $\mathbb{E} X = 1$ (otherwise we can scale $X$), and note that for all $a \in [0,\infty)$ and $k \geq 1$, $a^k-|a-1|^k > a-1$ and thus
\[
\mathbb{E}(X^k - |X-1|^k) \geq \mathbb{E}(X-1) = 0. \qedhere
\popQED
\]
\end{proof}
%%%%%%%%%%%%%%%%%%%%%%
Lemma  \ref{lemma:sub-exp-properties}  will enable us to derive the desired bounds on the sub-exponential norm $\subexpnorm{\cdot}$ of the random variables
$  \bar{N}_{s^K}(m)$ and $ \frac{\|H_k \|_2^2}{2 M}$. This is the content of the following lemma.
\end{shownto} %shownto arxiv

\begin{shownto}{conference}
Our arguments in the following sections will be based on a version of Bernstein's inequality for sub-exponential random variables from \cite{buldygin}. In order to apply the inequality to our problem, we need bounds on the sub-exponential norms of the random variables involved, which the following lemma provides. The proofs of both the lemma and the version of Bernstein's inequality that is used here are omitted due to lack of space and can be found in the extended version of this work~\cite{arxivpaper}.
\end{shownto}

%%%%%%%%%%%%%%%%%%%%%%%%%%%%%%%%%%%%%
%
%                          Lemma: Sub-exponential bound on fading and noise
%
%%%%%%%%%%%%%%%%%%%%%%%%%%%%%%%%%%%%%%

\begin{lemma}\label{lemma:key-lemma}
1. The random variables $(H^r_k (m))^2-1$ and $(H^i_k (m))^2-1$  are sub-exponential and we have
\begin{equation}\label{eq:key-lemma-1}
\subexpnorm{ H^r_k (m)^2 -1},\subexpnorm{ H^i_k (m)^2 -1}  \le 2\sigma_F^2,
\end{equation}
for all $k=1,\ldots,K$ and $m=1,\ldots ,M$.\\
2. For every $s^K\in\mathcal{S}_1 \times \ldots \times \mathcal{S}_K$ the random variables    $\bar{N}_{s^K}(m)$, $m=1,\ldots ,M$,  are sub-exponential and we have
\begin{IEEEeqnarray}{rCl}\label{eq:key-lemma-2}
\subexpnorm{   \bar{N}_{s^K}(m)- \mathbb{E}(\bar{N}_{s^K}(m))       } & \le & 6\sigma_F^2 \Delta(f\| P) \nonumber\\ 
& & + 8\sigma_N \sigma_F\sqrt{ \Delta(f\| P)}+ 4\sigma_{N}^2, \nonumber\\
\end{IEEEeqnarray}
where $\Delta(f\| P)$  is given in (\ref{eq:relative-spread}).
\end{lemma}
%%%%%%%%%%%%%%%%%%%%%%%%%%%%%%%%%%%%%%%%%%

\begin{shownto}{arxiv}
\begin{proof}
1.  The first claim follows easily from (\ref{eq:bound-fading}), (\ref{eq:c-s}), and the centering property (\ref{eq:centering}).

2.  We write $a_k:= g_k(f_k(s_k)) $ for $k=1,\ldots, K$ and note that from  (\ref{eq:pre-processing}), (\ref{eq:phi-def-spread}), (\ref{eq:relative-spread}), and (\ref{eq:max-inner-spread}) it follows that
\[\sum_{k=1}^K a_k \leq \Delta (f\| P).\]
First, we observe
\begin{equation}\label{eq:noisebounds-rotationalinvariance-prerequisite}
\subgaussnorm{\sqrt{a_k} H_k^i(m) U_k(m) } \leq \sqrt{a_k} \subgaussnorm{ H_k^i(m) } \leq \sqrt{a_k} \sigma_F,
\end{equation}
where the inequalities follow by (\ref{eq:bounded-12-bound}) and (\ref{eq:bound-fading}).
Since $\sqrt{a_k} H_k^i(m) U_k(m)$ have zero mean and are independent for $k = 1 \ldots K$, we can use (\ref{eq:subgauss-rotinv}) and get
\begin{equation}\label{eq:noisebounds-rotationalinvariance}
\subgaussnorm{ \sum_{k=1}^K \sqrt{a_k} H_k^i(m) U_k(m) } \leq \sqrt{\sum_{k=1}^K a_k \sigma_F^2} \leq \sigma_F \sqrt{ \Delta(f\| P)}.
\end{equation}
In the same fashion, we can derive the bound
\begin{equation}\label{eq:noisebounds-rotationalinvariance-real}
\subgaussnorm{ \sum_{k=1}^K \sqrt{a_k} H_k^r(m) U_k(m) } \leq \sqrt{\sum_{k=1}^K a_k \sigma_F^2} \leq \sigma_F \sqrt{\Delta(f\| P)}.
\end{equation}

We next bound the sub-exponential norm of the first summand in (\ref{eq:N-s-k}). Using the triangle inequality, (\ref{eq:c-s}), (\ref{eq:noisebounds-rotationalinvariance}), (\ref{eq:noisebounds-rotationalinvariance-real}) and (\ref{eq:bound-fading}), we obtain
\begin{align}
\nonumber
&\begin{aligned}
\hphantom{\le}~\subexpnorm{ \sum_{\substack{k ,l=1, \\ k\neq l}}^K \sqrt{a_k a_l} H_k(m)\overline{H_l (m)}  U_{k}(m)U_{l}(m)  }
\end{aligned}
\\
\nonumber
&\begin{aligned}
\le~ &\subexpnorm{ \sum_{\substack{k ,l=1}}^K \sqrt{a_k a_l} H_k(m)\overline{H_l (m)}  U_{k}(m)U_{l}(m)  }
\\
&+ \subexpnorm{ \sum_{\substack{k=1}}^K a_k |H_k(m)|^2 }
\end{aligned}
\displaybreak[0]
\\
\nonumber
&\begin{aligned}
\le~ &\subexpnorm{ \sum_{\substack{k ,l=1}}^K \sqrt{a_k a_l} H_k^i(m)H_l^i (m)  U_{k}(m)U_{l}(m)  }
\\
&+ \subexpnorm{ \sum_{\substack{k ,l=1}}^K \sqrt{a_k a_l} H_k^r(m)H_l^r (m)  U_{k}(m)U_{l}(m)  }
\\
&+ \subexpnorm{  \sum_{\substack{k=1}}^K a_k |H_k(m)|^2 }
\end{aligned}
\displaybreak[0]
\\
\nonumber
&\begin{aligned}
\le~ &2 \subgaussnorm{ \sum_{\substack{k=1}}^K \sqrt{a_k} H_k^i(m) U_{k}(m) }^2
\\
&+ 2 \subgaussnorm{ \sum_{\substack{k=1}}^K \sqrt{a_k} H_k^r(m) U_{k}(m) }^2 
\\
&+     \sum_{\substack{k=1}}^K a_k \left(\subgaussnorm{ H_k^i(m) }^2+\subgaussnorm{ H_k^r(m) }^2\right)
\end{aligned}
\\
\label{eq:first-term}
&\le~ 6 \sigma_F^2 \Delta(f\| P) .
\end{align}
For the second summand in (\ref{eq:N-s-k}) we use the triangle inequality, (\ref{eq:c-s}), (\ref{eq:noisebounds-rotationalinvariance}) and (\ref{eq:noisebounds-rotationalinvariance-real}) and obtain
\begin{align}
\nonumber
&\begin{aligned}
\hphantom{\le}~
\subexpnorm{ 2 \Real \left( \overline{N(m)} \sum_{k=1}^K \sqrt{a_k} H_k(m)U_{k}(m) \right ) }
\end{aligned}
\\
\nonumber
&\begin{aligned}
\le~ &2 \subexpnorm{  N^r(m)\sum_{k=1}^K \sqrt{a_k} H_k^r(m)U_{k}(m) }
\\ &+ 2 \subexpnorm{  N^i(m)\sum_{k=1}^K \sqrt{a_k} H_k^i(m)U_{k}(m) }
\end{aligned}
\displaybreak[0]
\\
\nonumber
&\begin{aligned}
\le~ &4 \subgaussnorm{ N^r(m) } \subgaussnorm{ \sum_{k=1}^K \sqrt{a_k} H_k^r(m)U_{k}(m) }
\\ &+   4 \subgaussnorm{ N^i(m) } \subgaussnorm{ \sum_{k=1}^K \sqrt{a_k} H_k^i(m)U_{k}(m) }
\end{aligned}
\\
\label{eq:second-term}
&\begin{aligned}
\le~ 8 \sigma_N \sigma_F \sqrt{\Delta(f\| P)}.
\end{aligned}
\end{align}
The norm of the last summand in (\ref{eq:N-s-k}) can be  bounded as follows:
\begin{IEEEeqnarray}{rCl}\label{eq:third-term}
\subexpnorm{ | N(m)     |^2   } &=& \subexpnorm{ N^r (m)^2 + N^i (m)^2   }\nonumber\\
&\le&  \subexpnorm{ N^r (m)^2 }  +  \subexpnorm{ N^i (m)^2  }  \nonumber\\
&\le&  2\subgaussnorm{  N^r (m)  }^2  +  2\subgaussnorm{  N^i (m)  }^2  \nonumber\\
&\le& 4 \sigma_{N}^2,
\end{IEEEeqnarray}
where in the second line we have used the triangle inequality, in the  third line we have used (\ref{eq:c-s}), and in the last line we have used (\ref{eq:bound-noise}).
Finally, the triangle inequality combined with the centering property (\ref{eq:centering}) applied to $| N(m)     |^2$, which is the only noncentered summand in (\ref{eq:N-s-k}), proves (\ref{eq:key-lemma-2}).
\end{proof}
Our arguments in the following sections will be based on a version of Bernstein's inequality for sub-exponential random variables from \cite{buldygin}.
\end{shownto} % shownto arxiv
%%%%%%%%%%%%%%%%%%%%%%%%%%%%%%%
%
%                     Bernstein's Inequality
%
%%%%%%%%%%%%%%%%%%%%%%%%%%%%%%%%%
\begin{theorem}[Bernstein's inequality]\label{theorem:bernstein}
Let $X_1, \ldots , X_M$ be independent and centered, and assume $\subexpnorm{X_1}, \dots \subexpnorm{X_M} \leq L$. Then for every $t\ge 0$, we have
\begin{equation}\label{eq:bernstein}
\mathbb{P}\left( \left |  \sum_{i=1}^M X_i    \right | \ge t     \right) \le 2 \exp \left ( -\frac{t^2}{2(Lt+2 M L^2)} \right).
\end{equation}
\end{theorem}
\begin{shownto}{arxiv}
\begin{proof}
The proof is given in Appendix \ref{appendix-bernstein}.
\end{proof}
\end{shownto}

\showto{conference}{The proof is a variation of the one given in~\cite{buldygin} and can be found in the extended version of this paper~\cite{arxivpaper}.}
%%%%%%%%%%%%%%%%%%%%%%%%%%%%%

%%%%%%%%%%%%%%%%%%%%%%%%%%%%%%%%%%%%%%%%%%%%%%%%%%%%%%%%%%%%%%%%%%%
%
%			Performance bounds
%
%%%%%%%%%%%%%%%%%%%%%%%%%%%%%%%%%%%%%%%%%%%%%%%%%%%%%%%%%%%%%%%%%%%
\subsection{Performance Bounds}\label{sec:performance-bounds}
The final step of the proof consists of bounding the probability of the error event
\begin{equation}
\mathbb{P}\left (    |   \bar{f} (s^K)- f(s^K)| \ge  \eps   \right )
\end{equation}
via Bernstein's inequality, Theorem \ref{theorem:bernstein}.\\
To this end, we use (\ref {eq:error-event-3}) which states that
\begin{multline}\label{eq:error-event-4}
\mathbb{P} \left (  |    \bar{f} (s^K)- f(s^K)       |\ge \eps    \right ) 
\\
\begin{aligned}
\le~
&\mathbb{P} \left( \left |     \sum_{k=1}^K f'_k(s_k) \left(\frac{\|H_k \|_2^2}{2 M}-1\right)     \right | \ge \eta   \     \right)
\\ &+ \mathbb{P}\left (  \left |    \frac{\alpha}{M}(\bar{N}_{s^K}- \mathbb{E}(\bar{N}_{s^K}))    \right | \ge     \eta     \     \right),
\end{aligned}
\end{multline}
where $\eta=\eta(\Phi, \eps):= \frac{\Phi^{-1}(\eps)}{2} $.\\
We bound the last term in (\ref{eq:error-event-4}) first. To this end we note that $\bar{N}_{s^K}= \sum_{m=1}^M  \bar{N}_{s^K}(m) $ is a sum of independent random variables by our construction,
and that we have the upper bound (\ref{eq:key-lemma-2}) for the sub-exponential norm of the random variables $ \bar{N}_{s^K}(m)$.
An application of Bernstein's inequality leads to
\begin{multline}\label{eq: prob-error-1}
 \mathbb{P}\left (  \left |    \frac{\alpha}{M}(\bar{N}_{s^K}- \mathbb{E}(\bar{N}_{s^K}))    \right | \ge     \eta     \     \right)
 \\ \le 2 \exp \left (   -\frac{M \eta^2}{2 L \eta + 4  L^2}   \right ),
\end{multline}
where 
\begin{equation}
L= 3\sigma_F^2 \bar{\Delta}(f)+\frac{4\sigma_N\sigma_F\sqrt{\Delta(f)\bar{\Delta}(f)}  }{\sqrt{P}}+\frac{2\sigma_{N}^2 \Delta(f)}{P}.
\end{equation}

Next, we use (\ref{eq:max-inner-spread}), (\ref{eq:key-lemma-1}) and Theorem~\ref{theorem:bernstein} to bound
\begin{align}
\nonumber
&\begin{aligned}
\hphantom{\le}~
\mathbb{P} \left( \Bigg |     \sum_{k=1}^K f'_k(s_k) \left(\frac{\|H_k \|_2^2}{2 M}-1 \right)      \Bigg | \ge \eta   \right)
\end{aligned}
\\
\nonumber
&\begin{aligned}
\le
\mathbb{P} \Bigg( \Bigg |     &\sum_{k,m=1}^{K, M } \frac{\Delta(f)}{2M} (((H_k^r(m))^2-1) \\&+ ((H_k^i(m))^2-1)) \Bigg |\ge \eta   \Bigg)
\end{aligned}
\\
\label{eq:prob-error-2}
&\begin{aligned}
\le
2 \exp \left( -\frac{M \eta^2}{2 \Delta(f) \sigma_F^2 \eta + 8 \Delta(f)^2 K \sigma_F^4} \right)
\end{aligned}
\end{align}
Combining (\ref{eq:prob-error-2}), (\ref{eq: prob-error-1}), and (\ref{eq:error-event-4}) concludes the proof of Theorem \ref{th:approximation-of-functions}.

%%%%%%%%%%%%%%%%%%%%%%%%%%%%%%%%%%%%%%%%%%%%%%%%%%%%%%%%%%%%%%%%%%%
%
%.                                 APPENDIX
%
%%%%%%%%%%%%%%%%%%%%%%%%%%%%%%%%%%%%%%%%%%%%%%%%%%%%%%%%%%%%%%%%%%%
\begin{shownto}{arxiv}
\appendix
\subsection{Proof of Bernstein's Inequality, Theorem \ref{theorem:bernstein}}\label{appendix-bernstein}
The proof is along the lines of the proof of \cite[Theorem 1.5.2]{buldygin}. We carry out the changes that are necessary 
to pass from the statement involving the second moments in  \cite[Theorem 1.5.2]{buldygin}  to the sub-exponential bounds in Bernstein's inequality.
%%%%%%%%%%%%%%%%%%%%%%%%%%%%%%%%%%%%%%%%%%%%%%%%%%%
%
%
%
%%%%%%%%%%%%%%%%%%%%%%%%%%%%%%%%%%%%%%%%%%%%%%%%%
\begin{lemma}\label{lemma:app-1}
Let $X$ be a random variable with $\E(X)=0$ and $\subexpnorm{X}< +\infty$. For any $\lambda \in\reals$ with $| \lambda \subexpnorm{X} |<1$ we have
\begin{equation}
\E (\exp(\lambda X))\le 1+ |\lambda|^2 \subexpnorm{X}^2 \cdot \frac{1}{1- |\lambda \subexpnorm{X}|}.
\end{equation}
\end{lemma}
%%%%%%%%%%%%%%%%%%%%%%%%%%%%%%%%%%%%%%%%%%%%%%% 
\begin{proof} Let $\lambda \in \reals$ satisfy $| \lambda \subexpnorm{X} |<1$.
Then
\begin{IEEEeqnarray}{rCl}
\E (\exp(\lambda X))&=& 1+ \sum_{k=2}^{\infty} \frac{\lambda ^k \E(X^k)}{k!}\nonumber\\
&\le& 1+ \sum_{k=2}^{\infty} \frac{|\lambda| ^k \E(|X|^k)}{k!}\nonumber\\
&\le& 1+ \sum_{k=2}^{\infty} |\lambda|^k \subexpnorm{X}^k\nonumber\\
&=& 1+ |\lambda |^2 \subexpnorm{X}^2\left ( \sum_{k=0}^{\infty} |\lambda \subexpnorm{X}  |^k \right)\nonumber\\
&=& 1+ |\lambda|^2 \subexpnorm{X}^2 \cdot \frac{1}{1- |\lambda \subexpnorm{X}|},
\end{IEEEeqnarray}
where in the last line we have used $| \lambda \subexpnorm{X} |<1$.
\end{proof}
In the next lemma we derive an exponential bound depending on $\subexpnorm{X}$ on the moment generating function of the random variable $X$.
%%%%%%%%%%%%%%%%%%%%%%%%%%%%%%%%%%%%%%%%%%%%%%%%%
%
%
%
%%%%%%%%%%%%%%%%%%%%%%%%%%%%%%%%%%%%%%%%%%%%%%%%%%
\begin{lemma}\label{lemma:app-2}
Let $X$ be a random variable with $\E(X)=0$ and $\subexpnorm{X}< +\infty$. For any $c\in (0,1)$ and 
$\lambda \in \left (  -\frac{c}{\subexpnorm{X}} ,  \frac{c}{\subexpnorm{X}} \right)$ we have
\begin{equation}
\E (\exp(\lambda X))\le \exp\left(\frac{\lambda^2}{2}\frac{2 \cdot  \subexpnorm{X}^2}{1-c}\right).
\end{equation}
\end{lemma}
%%%%%%%%%%%%%%%%%%%%%%%%%%%%%%%%%%%%%%%%%%%%%%%%%%%%
\begin{proof}
For $\lambda \in \left (  -\frac{c}{\subexpnorm{X}} ,  \frac{c}{\subexpnorm{X}} \right)$ we have
\begin{equation}\label{eq:app-1}
| \lambda \subexpnorm{X}    |< c<1,
\end{equation}
therefore by Lemma \ref{lemma:app-1}
\begin{IEEEeqnarray}{rCl}
\E (\exp(\lambda X))&\le& 1+ |\lambda|^2 \subexpnorm{X}^2 \cdot \frac{1}{1- |\lambda \subexpnorm{X}|}\nonumber\\
&\le&  1+ |\lambda|^2 \subexpnorm{X}^2 \cdot \frac{1}{1- c}\nonumber\\
&\le& \exp\left( \frac{\lambda^2}{2}\frac{2 \cdot \subexpnorm{X}^2}{1-c}     \right),
\end{IEEEeqnarray}
where in the second line we have used the first inequality in (\ref{eq:app-1})  and the last line is by the numerical inequality $1+x \le \exp (x)$ valid for $x\ge 0$.
\end{proof}
%%%%%%%%%%%%%%%%%%%%%%%%%%%%%%%%%%%%%%%%%%%%%%%%%%
%
%
%
%%%%%%%%%%%%%%%%%%%%%%%%%%%%%%%%%%%%%%%%%%%%%%%%%
\begin{lemma}\label{lemma:app-3}
Let $X_1, \ldots , X_M$ be independent random variables with $\E(X_i)=0$ and $\subexpnorm{X_i}<+\infty$, $i=1,\ldots ,M$. 
Let $L:=\max_{1\le i\le M} \subexpnorm{X_i}$, $c\in (0,1)$, and $\lambda \in \left (  -\frac{c}{L} ,  \frac{c}{L} \right)$.
Then for $S_M:= \sum_{i=1}^M X_i$ we have
\begin{equation}\label{eq:app-2}
\E (\exp(\lambda S_M))\le \exp \left (  \frac{\lambda^2}{2}\frac{2 \cdot \sum_{i=1}^M \subexpnorm{X_i}^2}{1-c}  \right).
\end{equation}
\end{lemma}
%%%%%%%%%%%%%%%%%%%%%%%%%%%%%%%%%%%%%%%%%%%%%
\begin{proof}
By independence of $X_1, \ldots , X_M$ we have
\begin{equation}
\E (\exp(\lambda S_M))= \prod_{i=1}^M \E (\exp(\lambda X_i)).
\end{equation}
Combining this with Lemma \ref{lemma:app-2} proves the claim of the lemma.
\end{proof}
%%%%%%%%%%%%%%%%%%%%%%%%%%%%%%%%%%%%%%%%%%%%%%
The next lemma establishes the basic tail bound for random variables satisfying inequalities of type (\ref{eq:app-2}).
The proof can be found in \cite[Lemma 1.4.1]{buldygin}.
%%%%%%%%%%%%%%%%%%%%%%%%%%%%%%%%%%%%%%%%%%%%%%%%%%
%
%
%
%%%%%%%%%%%%%%%%%%%%%%%%%%%%%%%%%%%%%%%%%%%%%%%
\begin{lemma}\label{lemma:app-4}
Let $X$ be a random variable with $\E(X)=0$. If there exist $\tau\ge 0$ and $\Lambda >0$ such that
\begin{equation}
\E (\exp(\lambda X))\le \exp\left(\frac{\lambda^2}{2}\tau^2\right),
\end{equation}
holds for all $\lambda\in (-\Lambda, \Lambda)$,
then for any $t\ge 0$ we have
\begin{equation}
\mathbb{P}(|X |\ge t)\le 2 \cdot Q(t),
\end{equation}
where
\begin{equation}
Q(t)=  
\begin{cases}   \exp \left (- \frac{t^2}{2 \tau^2}  \right ), & 0< t\le \Lambda \tau^2 \\
         \exp\left (     - \frac{\Lambda t}{2}    \right), & \Lambda  \tau^2\le t .
\end{cases}
\end{equation}
\end{lemma}
%%%%%%%%%%%%%%%%%%%%%%%%%%%%%%%%%%%%%%%%%%%%%%%
Lemma \ref{lemma:app-4} and Lemma \ref{lemma:app-3} immediately imply the following tail inequality for sums of independent sub-exponential random variables.
%%%%%%%%%%%%%%%%%%%%%%%%%%%%%%%%%%%%%%%%%%%%%%%
%
%
%
%%%%%%%%%%%%%%%%%%%%%%%%%%%%%%%%%%%%%%%%%%%%%%%%
\begin{lemma}\label{lemma:app-5}
Let $X_1, \ldots , X_M$ be independent random variables with $\E(X_i)=0$ and $\subexpnorm{X_i}<+\infty$ for $k=1,\ldots, M$.
Let $L:= \max_{1\le i \le M} \subexpnorm{X_i}$, $c\in (0,1)$, and $S_M:=\sum_{i=1}^M X_i$.\\
Then for any $t\ge 0$
\begin{equation}
\mathbb{P}(|S_M    |\ge t)\le 2 \cdot V(c,t),
\end{equation}
where
\begin{equation}\label{eq:app-3}
V(c,t)=
\begin{cases}  \exp \left (- \frac{t^2 (1-c)}{4  \sum_{i=1}^M \subexpnorm{X_i}^2}  \right ), & 0< t\le \frac{2 c \sum_{i=1}^M \subexpnorm{X_i}^2}{L (1-c)} \\
                          \exp\left (     - \frac{c t}{2 L}    \right), &   \frac{2 c \sum_{i=1}^M \subexpnorm{X_i}^2}{L (1-c)}  \le t .
\end{cases}
\end{equation}
\end{lemma}
%%%%%%%%%%%%%%%%%%%%%%%%%%%%%%%%%%%%%%%%%%%%%%
The proof of Theorem \ref{theorem:bernstein} follows from Lemma \ref{lemma:app-5}. To this end we observe that for fixed $t\ge 0$ the lower expression in (\ref{eq:app-3}) is decreasing in $c$ while the upper expression in (\ref{eq:app-3}) is increasing in $c$.
The $c'$ that minimizes the $V(c',t)$ for fixed $t\ge 0$ can be found by solving the equation
\begin{equation}
 \frac{2 c' \sum_{i=1}^M \subexpnorm{X_i}^2}{L (1-c')} =  t,
\end{equation}
resulting in 
\begin{equation}\label{eq:app-4}
c'=\frac{L t}{L t + 2 \sum_{i=1}^M \subexpnorm{X_i}^2}.
\end{equation}
Inserting the expression for $c'$ given in (\ref{eq:app-4}) into $V(c',t)$ and using the representation (\ref{eq:app-3}) we obtain
\begin{equation}\label{eq:app-5}
\mathbb{P}(   |S_M    |\ge t ) \le 2 \cdot \exp\left (- \frac{t^2}{2 (L t + 2 M L^2)}     \right ),
\end{equation}
where we have used
\begin{equation}
\sum_{i=1}^M \subexpnorm{X_i}^2\le M L^2.
\end{equation}

\end{shownto} % shownto arxiv

\bibliography{references}

\begin{thebibliography}{10}

\bibitem{agrawal2019scalable}
N.~Agrawal, M.~Frey, and S.~Stanczak.
\newblock A scalable max-consensus protocol for noisy ultra-dense networks.
\newblock In {\em 2019 IEEE 20th International Workshop on Signal Processing
  Advances in Wireless Communications (SPAWC)}. IEEE, 2019.

\bibitem{buck1976approximate}
R.~C. Buck.
\newblock Approximate complexity and functional representation.
\newblock Technical report, Wisconsin University Madison Mathematics Research
  Center, 1976.

\bibitem{buck1982nomographic}
R.~C. Buck.
\newblock Nomographic functions are nowhere dense.
\newblock {\em Proceedings of the American Mathematical Society}, pages
  195--199, 1982.

\bibitem{buldygin}
V.V. Buldygin and Yu.V. Kozachenko.
\newblock {\em Metric Characterization of Random Variables and Random
  Processes}.
\newblock Cross Cultural Communication. American Mathematical Soc., 2000.

\bibitem{christmann2012consistency}
A.~Christmann and R.~Hable.
\newblock Consistency of support vector machines using additive kernels for
  additive models.
\newblock {\em Computational Statistics \& Data Analysis}, 56(4):854--873,
  2012.

\bibitem{goldenbaum2013harnessing}
M.~Goldenbaum, H.~Boche, and S.~Sta{\'n}czak.
\newblock Harnessing interference for analog function computation in wireless
  sensor networks.
\newblock {\em IEEE Transactions on Signal Processing}, 61(20):4893--4906,
  2013.

\bibitem{goldenbaum2014nomographic}
M.~Goldenbaum, H.~Boche, and S.~Sta{\'n}czak.
\newblock Nomographic functions: Efficient computation in clustered gaussian
  sensor networks.
\newblock {\em IEEE Transactions on Wireless Communications}, 14(4):2093--2105,
  2014.

\bibitem{goldenbaum2016harnessing}
M.~Goldenbaum, P.~Jung, M.~Raceala-Motoc, J.~Schreck, S.~Sta{\'n}czak, and
  C.~Zhou.
\newblock Harnessing channel collisions for efficient massive access in 5{G}
  networks: A step forward to practical implementation.
\newblock In {\em 2016 9th International Symposium on Turbo Codes and Iterative
  Information Processing (ISTC)}, pages 335--339. IEEE, 2016.

\bibitem{goldenbaum2013robust}
M.~Goldenbaum and S.~Stanczak.
\newblock Robust analog function computation via wireless multiple-access
  channels.
\newblock {\em IEEE Transactions on Communications}, 61(9):3863--3877, 2013.

\bibitem{kolmogorov1957representation}
A.~N. Kolmogorov.
\newblock On the representation of continuous functions of several variables by
  superposition of continuous functions of one variable and addition.
\newblock {\em Dokl. Akad. Nauk SSS}, 114:953--956, 1957.

\bibitem{nazer2007computation}
B.~Nazer and M.~Gastpar.
\newblock Computation over multiple-access channels.
\newblock {\em IEEE Transactions on information theory}, 53(10):3498--3516,
  2007.

\bibitem{kiril}
K.~Ralinovski, M.~Goldenbaum, and S.~Sta{\'n}czak.
\newblock Energy-efficient classification for anomaly detection: The wireless
  channel as a helper.
\newblock In {\em 2016 IEEE International Conference on Communications (ICC)},
  pages 1--6, 2016.

\bibitem{steinwart}
I.~Steinwart and A.~Christmann.
\newblock {\em Support Vector Machines}.
\newblock Information Science and Statistics. Springer, 2008.

\bibitem{vershynin}
R.~Vershynin.
\newblock {\em High Dimensional Probability: An Introduction with Applications
  in Data Science}, volume~47 of {\em Cambridge Series in Statistical and
  Probabilistic Mathematics}.
\newblock Cambridge University Press, 2018.

\bibitem{wainwright}
Martin~J. Wainwright.
\newblock {\em High-Dimensional Statistics: A Non-Asymptotic Viewpoint}.
\newblock Cambridge Series in Statistical and Probabilistic Mathe\-matics.
  Cambridge University Press, 2019.

\end{thebibliography}
\bibliographystyle{plain}

\end{document}